\documentclass[a4paper,12pt]{article}

\usepackage{amsmath,amsthm,amssymb,amsfonts}

\makeatletter

\@addtoreset{equation}{section}
\makeatother

\newtheorem{theorem}{Theorem}[section]
\newtheorem{lemma}[theorem]{Lemma}

\newtheorem{corollary}[theorem]{Corollary}

\def\ep{\varepsilon}

\def\R{\mathbb R}
\def\S{\mathbb S}
\def\N{\mathbb N}
\def\pa{\partial}
\def\b{\backslash}
\def\A{\mathbf A}

\def\V{\mathcal V}
\def\B{\mathcal B}
\def\A{\mathcal A}
\def\G{\mathcal G}

\begin{document}

\title{On inverse scattering at high energies for the multidimensional relativistic Newton equation in a long range electromagnetic field}
\author{Alexandre Jollivet\footnote{Laboratoire de Math\'ematiques Paul Painlev\'e,
CNRS UMR 8524/Universit\'e Lille 1 Sciences et Technologies,
59655 Villeneuve d'Ascq Cedex, France,\
alexandre.jollivet@math.univ-lille1.fr}}

\maketitle

\begin{abstract} 
We define scattering data for the relativistic Newton equation in an electric field  $-\nabla V\in C^1(\R^n,\R^n)$, $n\ge 2$, and in a magnetic field $B\in C^1(\R^n,A_n(\R))$ that decay at infinity like $r^{-\alpha-1}$
for some $\alpha\in (0,1]$, where $A_n(\R)$ is the space of $n\times n$ antisymmetric matrices.
We provide estimates on the scattering solutions and on the scattering data and we prove, in particular, that the scattering data at high energies uniquely determine the short range part of $(\nabla V,B)$ up to the knowledge of the long range tail of $(\nabla V,B)$.
The Born approximation at fixed energy of the scattering data is also considered. We then change the definition of the scattering data to study their behavior in other asymptotic regimes. 
This work generalizes [Jollivet, 2007] where a short range electromagnetic field was considered.
\end{abstract}

\section{Introduction}
Consider the multidimensional relativistic Newton equation in an external and static electromagnetic field:
\begin{eqnarray}
&&\dot p(t) = F(x(t),\dot x(t)):=-\nabla V(x(t))+{1\over c}B(x(t))\dot x(t),\label{1.1}\\
&&\ p(t)={\dot x(t) \over \sqrt{1-{|\dot x(t)|^2 \over c^2}}},\ x(t)\in \R ^n,\ t\in\R,\ n\ge2,\nonumber
\end{eqnarray}
where $\dot x(t)={d\over dt}x(t)$, and where $V \in C^2(\R^n,\R),$ $B(x)$ is the $n\times n$ real antisymmetric matrix with elements $B_{i,k}$, $1\le i,k\le n$, and where $B$ satisfies the closure condition
\begin{equation}
{\pa\over \pa x_i}B_{k,m}(x)+{\pa\over \pa x_m}B_{i,k}(x)+{\pa\over \pa x_k}B_{m,i}(x)=0,\label{1.1b}
\end{equation}
for $x\in \R^n$ and for $i,k,m=1\ldots n$.
The constant $c$ is positive, and for $\sigma\in(0,+\infty)$ we will denote by $\B(0,\sigma)$ (resp. $\overline{\B(0,\sigma)}$) the open (resp. closed) Euclidean Ball of center 0 and radius $\sigma$. 

When $n=3$ the equation \eqref{1.1} is the equation of motion of a relativistic particle of mass $m=1$ and charge $e=1$ in an
external electromagnetic field described by $(V,B)$ (see \cite{[E]} and, for example, \cite[Section 17]{[LL2]}). In this equation, $x$, $\dot x$, $p$ denote the position, the velocity and the impulse of the particle respectively, and $t$ is the time, and $c$ is the speed of light.

We also assume throughout this paper that $F$ satisfies the following conditions
\begin{equation}
F=F^l+F^s,\label{1.3}
\end{equation}
where $F^l(x,v):=-\nabla V^l(x)+{1\over c}B^l(x)v$, $F^s(x,v)=-\nabla V^s(x)+{1\over c}B^s(x)v$ and $(V^l,V^s)\in (C^2(\R^n,\R))^2$, $(B^l,B^s)\in (C^1(\R^n,A_n(\R)))^2$, and where
\begin{equation}
|\pa^{j_1}_x V^l(x)| \le \beta_{|j_1|}^l(1+|x|)^{-(\alpha+|j_1|)},\ |\pa^{j_2}_x B_{i,k}^l(x)| \le \beta_{|j_2|+1}^l(1+|x|)^{-(\alpha+|j_2|+1)},\label{1.4a}
\end{equation}
\begin{equation}
|\pa^{j_1}_x V^s(x)| \le \beta_{|j_1|+1}^s(1+|x|)^{-(\alpha+1+|j_1|)},\ |\pa^{j_2}_x B_{i,k}^s(x)| \le \beta_{|j_2|+2}^s(1+|x|)^{-(\alpha+|j_2|+2)},\label{1.4b}
\end{equation}
for $x\in \R^n$, $|j_1| \le 2$ and $|j_2|\le 1$ and for some $\alpha\in (0,1]$ (here $j$ is the multiindex $j=(j^1,\ldots,j^n)\in (\N \cup \{0\})^n, |j|= \sum_{m=1}^n j^m$, and $\beta_m^l$ and $\beta_{m'}^s$ are positive real constants for $m=0,1,2$ and for $m'=1,2,3$). Note that the assumption $0<\alpha\le 1$ includes the decay rate of a Coulombian potential at infinity. Indeed for a Coulombian potential $V^l(x)={1\over |x|}$ (and no magnetic field $B ^l\equiv0$), estimates \eqref{1.4a} are satisfied uniformly for $|x|>\ep$ and $\alpha=1$ for any $\ep>0$. Although our electromagnetic fields are assumed to be smooth on the entire space, our study may provide interesting results even in presence of singularities.

For equation \eqref{1.1} the energy 
\begin{equation}
E=c^2\sqrt{1+{|p(t)|^2 \over c^2}}+V(x(t))\label{1.2}
\end{equation}
is an integral of motion.

For $\sigma\in[0,+\infty)$ set 
\begin{equation}
\mu(\sigma)=\sqrt{2\sigma\over {\sigma\over c^2}+\sqrt{{\sigma^2\over c^4}+4}}\label{1.2b}
\end{equation}
and $\mu^l=\mu\big(2^8\alpha^{-1}n^2\max(\beta_1^l,\beta_2^l)\big)$.
Then under conditions \eqref{1.4a} the following is valid (see Lemma \ref{lem_scatinit} given in the next Section): for any $v\in \B(0,c)$, $|v|\ge \mu^l$, there exists a unique solution $z_\pm(v,.)$ of the equation
\begin{eqnarray}
\dot p(t)&=&F^l(z(t),\dot z(t)),\label{5.3}\\
p(t)&=&{\dot z(t) \over \sqrt{1-{|\dot z(t)|^2 \over c^2}}},\ t\in \R,\nonumber
\end{eqnarray}
so that
\begin{equation*}
\dot z_\pm(v,t)-v=o(1),\textrm{ as }t\to \pm \infty,\ z_\pm(v,0)=0,
\end{equation*}
and 
\begin{equation*}
|z_\pm(v,t)-tv|\le {2^{9\over 2}n^{3\over 2}\beta_1^l\sqrt{1-{|v|^2\over c^2}}\over \alpha |v|}|t|\textrm{ for }t\in \R.
\end{equation*}
When $F^l\equiv0$ then  $\beta_1^l$, $\beta_2^l$ and $\mu^l$ can be arbitrary close to $0$, and we have $z_\pm(v,t)=tv$ for $(t,v)\in \R\times \B(0,c)$, $v\not=0$.

Then 
under conditions \eqref{1.4a} and \eqref{1.4b}, the following is valid: for any 
$(v_-,x_-)\in \B(0,c)\b\B(0,\mu^l)\times\R^n,$ 
the equation \eqref{1.1}  has a unique solution $x\in C^2(\R,\R^n)$ such that
\begin{equation}
{x(t)=z_-(v_-,t)+x_-+y_-(t),}\label{1.6}
\end{equation}
where $|\dot y_-(t)|+|y_-(t)|\to 0,$ as $t\to -\infty;$  in addition for almost any 
$(v_-,x_-)\in \B(0,c)\b\B(0,\mu^l)\times \R^n$,
\begin{equation}
{x(t)=z_+(v_+,t)+x_++y_+(t),}\label{1.7}
\end{equation}
for a unique $(v_+,x_+)\in \B(0,c)\times\R^n$, where $|v_+|=|v_-|\ge \mu^l$ by conservation of the energy \eqref{1.2}, and where $v_+=:a(v_-,x_-)$, $x_+=:b(v_-,x_-),$ and  
$|\dot y_+(t)|+|y_+(t)|\to 0,$ as $t \to +\infty$. A solution $x$ of \eqref{1.1} that satisfies \eqref{1.6} and \eqref{1.7} for some $(v_-,x_-)$, $v_-\not=0$, is called a scattering solution. 

We call the map 
$S: (\B(0,c)\b \B(0,\mu^l))\times\R^n \to (\B(0,c)\b\B(0,\mu^l))\times\R^n$
given by the formulas
\begin{equation}
{v_+=a(v_-,x_-),\ x_+=b(v_-,x_-)},\label{1.8}
\end{equation}
the scattering map for the equation \eqref{1.1}. In addition, $a(v_-,x_-),$ $b(v_-,x_-)$ are called the scattering data for
the equation \eqref{1.1}, and we define
\begin{equation}
a_{sc}(v_-,x_-)=a(v_-,x_-)-v_-,\ \ b_{sc}(v_-,x_-)=b(v_-,x_-)-x_-.\label{1.10}
\end{equation}
Our definition of the scattering map is derived from constructions given in \cite{He, DG}. We refer the reader to \cite{He, DG} and references therein for the forward classical scattering theory.

By ${\cal D}(S)$ we denote the set of definition of $S$. Under the conditions \eqref{1.4a} and \eqref{1.4b} the map  $S:{\cal D}(S)\to (\B(0,c)\b \B(0,\mu^l))\times\R^n$ is continuous, 
and ${\rm Mes}(((\B(0,c)\b \B(0,\mu^l))\times \R^n) \b {\cal D}(S))=0$ for the Lebesgue
measure on $\B(0,c) \times \R^n$. 
In addition the map $S$ is uniquely determined by its restriction to ${\cal M}(S)={\cal
D}(S)\cap {\cal M}$ and by $F^l$, where
${\cal M}=\{(v_-,x_-)\in \B(0,c) \times \R^n\ |\ v_-\neq 0, v_-\cdot x_-=0\}$.
(Indeed if $x(t)$ is a solution of \eqref{1.1} then $x(t+t_0)$ is also a solution of \eqref{1.1} for any $t_0\in \R$.)

One can imagine the following experimental setting that allows to measure the scattering data without knowing the electromagnetic field $(V,B)$ inside a (a priori bounded) region of interest.
First find an electromagnetic field $(V^l,B^l)$ that generates the same long range effects as $(V,B)$ does. Then compute the solutions $z_\pm(v,.)$ of equation \eqref{5.3}. Then for a fixed $(v_-,x_-)\in (\B(0,c)\b \B(0,\mu^l))\times \R^n$ send a particle far away from the region of interest  with a trajectory asymptotic to $x_-+z_-(v_-,.)$ at large and negative times.  When the particle escapes any bounded region of the space at finite time, then detect the particle  and find $S(v_-,x_-)=(v_+,x_+)$ so that  the trajectory of the particle is  asymptotic to  $x_++z_+(v_+,.)$ at large and positive times far away from the bounded region of interest.

In this paper we consider the following inverse scattering problem for equation \eqref{1.1}:
\begin{equation}
\textrm{Given }S\textrm{ and given the long range tail }F^l\textrm{ of the force }F,\textrm{ find }F^s. \label{P}
\end{equation}

The main results of the present work consist in estimates and asymptotics for the scattering data $(a_{sc},b_{sc})$ 
and scattering solutions for the equation \eqref{1.1} and in application of these asymptotics and estimates to the inverse scattering problem \eqref{P} at high energies. Our main results include, in particular, Theorem \ref{thm} given below that provides the high energies asymptotics of the scattering data and the Born approximation of the scattering data at fixed energy.

Consider
\begin{equation*}
\label{1.9}
T\S^{n-1}:=\{(\theta,x)\in \S^{n-1}\times\R^n\ |\ \theta\cdot x=0\},
\end{equation*}
and for any $m\in \N$ consider the x-ray transform $P$ defined by
\begin{equation*}
Pf(\theta,x):=\int_{-\infty}^{+\infty}f(t\theta+x)dt
\end{equation*}
for any function $f\in C(\R^n,\R^m)$ so that $|f(x)|=O(|x|^{-\tilde \beta})$ as $|x|\to +\infty$ for some $\tilde \beta>1$. 
For $(\sigma,\tilde \beta,r,\tilde \alpha)\in (0,+\infty)^2\times(0,\min(1,2^{-{3\over 2}}c))\times (0,1]$, let $\rho_0=\rho_0(\sigma,r,\tilde \beta,\tilde \alpha)$ be defined as the root of the equation
\begin{equation}
1={32n^2\sqrt{1-{\rho_0^2\over c^2}}\tilde \beta(1+\sigma+{1\over c})(1+{1\over {\rho_0\over 2^{3\over 2}}-r})\over \tilde \alpha({\rho_0\over 2^{3\over 2}}-r)r(1-r)^{\tilde \alpha+2}},\ \rho_0\in(2^{3\over 2}r,c).\label{1.30}
\end{equation}
Set
\begin{eqnarray}
W(v,x)&:=&\int_{-\infty}^0\Big(g\big(g^{-1}(v)+\int_{-\infty}^\sigma F^l(z_-(v,\tau)+x,\dot z_-(v,\tau))d\tau\big)\label{10.10}\\
&&\hskip -1cm -g\big(g^{-1}(v)+\int_{-\infty}^\sigma F^l(z_-(v,\tau),\dot z_-(v,\tau))d\tau\big)\Big)d\sigma\nonumber\\
&&\hskip -1cm +\int_0^{+\infty}\Big(g\big(g^{-1}(a(v,x))-\int_\sigma^{+\infty} F^l(z_+(a(v,x),\tau)+x,\dot z_+(a(v,x),\tau))d\tau\big)\nonumber\\
&&\hskip -1cm -g\big(g^{-1}(a(v,x))-\int_\sigma^{+\infty} F^l(z_+(a(v,x),\tau),\dot z_+(a(v,x),\tau))d\tau\big)\Big)d\sigma,\nonumber
\end{eqnarray}
for $(v,x)\in {\cal D}(S)$, where 
\begin{equation}
g(x):={x\over \sqrt{1+{|x|^2\over c^2}}}\textrm{ for } x\in \R^n,\textrm{ and }
g^{-1}(x'):={x'\over \sqrt{1-{|x'|^2\over c^2}}}\textrm{ for }x'\in \B(0,c).\label{01}
\end{equation}
Then we have the following results.

\begin{theorem}
\label{thm}
Let $(\theta,x)\in T\S^{n-1}$. Under conditions \eqref{1.4a} and \eqref{1.4b} the following limits are valid
\begin{eqnarray}
\lim_{\rho\to c\atop \rho<c} {\rho\over \sqrt{1-{\rho^2\over c^2}}} a_{sc}(\rho\theta,x)&=&\int_{-\infty}^{+\infty}F(\tau \theta+x,c\theta)d\tau,\label{t1}\\
\lim_{\rho\to c\atop\rho<c}{\rho^2\over \sqrt{1-{\rho^2\over c^2}}}(b_{sc}(\rho\theta,x)-W(\rho\theta,x))&=&\int_{-\infty}^0\int_{-\infty}^\sigma F^s(\tau \theta+x,c\theta)d\tau d\sigma
\nonumber\\
&&\hskip -4cm-\int_0^{+\infty}\int_\sigma^{+\infty} F^s(\tau \theta+x,c\theta)d\tau d\sigma+PV^s(\theta,x)\theta.\label{t2}
\end{eqnarray}
In addition
\begin{eqnarray}
&&\hskip-1.5cm|{\rho\over\sqrt{1-{\rho^2\over c^2}}}a_{sc}(\rho\theta,x)-\int_{-\infty}^{+\infty}F(\tau \theta+x,\rho\theta)d\tau|\nonumber\\
&\le& \beta^2\sqrt{1-{\rho^2\over c^2}}{640n^4\rho
\big({r\over c}+|x|+1\big)({1\over c}+1)\big(1+{1\over {\rho\over 2\sqrt{2}}-r}\big)^2\over \alpha^2({\rho\over 2\sqrt{2}}-r)^2(1-r)^{2\alpha+3}}
,\label{t3b}
\end{eqnarray}
\begin{eqnarray}
&&\hskip-1.5cm\Big|{\rho^2\over\sqrt{1-{\rho^2\over c^2}}}\big(b_{sc}(\rho\theta,x)-W(\rho\theta,x)\big)-\int_{-\infty}^0\int_{-\infty}^\sigma F^s(\tau \theta+x,\rho\theta)d\tau d\sigma\nonumber\\
&&+\int_0^{+\infty}\int_\sigma^{+\infty} F^s(\tau \theta+x,\rho\theta)d\tau d\sigma-\int_{-\infty}^{+\infty}V^s(\sigma \theta+x)d\sigma {\rho^2\theta\over c^2}\Big|\nonumber\\
&\le&\beta^2\sqrt{1-{\rho^2\over c^2}}{464n^4\rho^2\big({r\over c}+|x|+1\big)({1\over c}+1)\big(1+{1\over  
{\rho\over 2\sqrt{2}}-r}\big)^2\over \alpha^2(\alpha+1)({\rho\over 2\sqrt{2}}-r)^3(1-r)^{2\alpha+2}},\label{t4b}
\end{eqnarray}
for $r\in (0,\min(1,2^{-{3\over 2}}c))$ and for $\rho\in(\rho_0(|x|,r,\beta,\alpha),c)$, 
where $\beta=\max(\beta_1^l,\beta_2^l,$ $\beta_2^s,\beta_3^s)$.
\end{theorem}

The vector $W$ defined by \eqref{10.10} is known from the scattering data and from $F^l$.
Then from \eqref{t1}  and \cite[Proposition 1.1]{Jo1} and inversion of the x-ray transform (see \cite{R, GGG, Na, No})  it follows that $F^s$ can be  reconstructed from  $a_{sc}$. 
From \eqref{t2} one can prove the following statements (see \cite[Proposition 1.2]{Jo1} and subsequent comments therein):
The potential $V^s$ is uniquely determined up to its radial part by $b_{sc}$; The magnetic field $B^s$ can be reconstructed from $b_{sc}$ when $n\ge 3$, and up to its radial part when $n=2$.

The estimates \eqref{t3b} and \eqref{t4b} also
give the asymptotics of  $a_{sc},$ $b_{sc}$, when the parameters $\alpha,$ $n,$ $\rho$, $\theta$ and $x$ are fixed and
$\beta$ decreases to $0$. In that regime the leading term of $a_{sc}(\rho\theta,x)$ and
$b_{sc}(\rho\theta,x)-W(\rho\theta,x)$
for $(\theta,x)\in T\S^{n-1}$ and for $\rho\in (\rho_0(|x|,r,\beta,\alpha),c)$ are given by 
\begin{eqnarray}
&&{\sqrt{1-{\rho^2\over c^2}}\over \rho}\int_{-\infty}^{+\infty}F(\tau \theta+x,\rho\theta)d\tau,\label{t3c}\\
&&{\sqrt{1-{\rho^2\over c^2}}\over \rho^2}\Big(\int_{-\infty}^0\int_{-\infty}^\sigma F^s(\tau \theta+x,\rho\theta)d\tau d\sigma\nonumber\\
&&-\int_0^{+\infty}\int_\sigma^{+\infty} F^s(\tau \theta+x,\rho\theta)d\tau d\sigma+\int_{-\infty}^{+\infty}V^s(\sigma \theta+x)d\sigma {\rho^2\theta\over c^2}\Big),\label{t4c}
\end{eqnarray}
respectively. 
Therefore Theorem \ref{thm} gives the Born approximation for the scattering data at fixed energy when
the electromagnetic field is sufficiently weak, and one can prove the following statements (see \cite[Remark 1.1]{Jo1}):
The force $F^s$ can be reconstructed from the Born approximation \eqref{t3c} of $a_{sc}$ at fixed energy;
$V^s$ can be reconstructed from the Born approximation \eqref{t4c} of $b_{sc}$ at fixed energy; 
$B^s$ can be reconstructed from \eqref{t4c} when $n\ge 3$, and up to its radial part when $n=2$.

Theorem \ref{thm} is a generalization of \cite[formulas (1.7a), (1.7b), (1.8a) and (1.8b)]{Jo1} where inverse scattering for the relativistic multidimensional Newton equation was studied in the short range case ($F^l\equiv 0$). 
The formulas \cite[(1.7b) and (1.8b)]{Jo1}  also provide the approximation of the scattering data
$(a_{sc}(v_-,x_-),b_{sc}(v_-,x_-))$ for the short range case ($F^l\equiv 0$) when the parameters $\alpha,$ $n,$ $v_-$ and $\beta$ are fixed and
$|x_-|\to +\infty$. Such an asymptotic regime is not covered by Theorem \ref{thm}. Therefore we shall modify in Section 3 the definition of the scattering map 
to study these modified scattering data in the following three asymptotic regimes: at high energies, Born approximation at fixed energy, and when the parameters $\alpha,$ $n$, $v_-$ and $\beta$ are fixed and
$|x_-|\to +\infty$.

Inverse scattering at high energies for the nonrelativistic multidimensional Newton equation in a short range potential $V$ was first studied by \cite{No}. Then inverse scattering at high energies for this latter equation in a long range potential $V$ was studied by \cite{Jo2}. We develop the approach of \cite{No, Jo2} to obtain our results.

For inverse scattering at fixed energy for the multidimensional Newton equation, see for example \cite{Jo3} and references therein.
For the inverse scattering problem in quantum mechanics, see references given in \cite{Jo1}.

Our paper is organized as follows. In Section 2 we transform the differential equation \eqref{1.1} with initial conditions \eqref{1.6} in an integral equation which takes the form $y_-=A(y_-).$ Then we study the operator $A$ on a
suitable space  (Lemma \ref{lem_cont}) and we give estimates for the deflection $y_-(t)$ in \eqref{1.6} and for  
the scattering data $a_{sc}(v_-,x_-),b_{sc}(v_-,x_-)$  (Theorem \ref{thm_y}). We prove Theorem \ref{thm}.
Note that we work with small angle scattering compared to the dynamics generated by $F^l$ through the ``free" solutions $z_-(v_-,t)$: In particular, the angle between the vectors $\dot x(t)=\dot z_-(v_-,t)+\dot y_-(t)$ and $\dot z_-(v_-,t)$ goes to zero
when the parameters $\beta$, $\alpha$, $n$, $v_-/|v_-|$, $x_-$ are fixed and $|v_-|$ increases.
In Section 3 we change the definition of the scattering map so that one can obtain for the modified scattering data $(\tilde a_{sc}(v_-,x_-),\tilde b_{sc}(v_-,x_-))$ their approximation  at high energies, or their Born approximation at fixed energy, or their approximation  when the parameters $\alpha,$ $n,$ $v_-$ and $\beta$ are fixed and
$|x_-|\to +\infty$ (Theorem \ref{thm_y2}, Corollary \ref{high_energy}). 
Sections 4, 5, 6 and 7 are devoted to proofs of our Theorems and Lemmas.

\section{Scattering solutions}
\subsection{Integral equation}
First we need the following Lemma \ref{lem_scatinit} that generalizes the statements given in the Introduction on the existence of peculiar solutions $z_\pm$ of the equation \eqref{5.3}.

\begin{lemma}
\label{lem_scatinit}
Assume conditions \eqref{1.4a}. Let $v\in \B(0,c)$, $v\not=0$, and $x\in \R^n$  so that $v\cdot x=0$. Let $(w,q)\in \B(0,c)\times\B(0,1)$ so that
\begin{equation}
|v|=|w|,\textrm{ and } |w-v|\le {|v|\over 2^{5\over 2}}.\label{5.1}
\end{equation}
Assume that
\begin{equation}
{2^8n^2\max(\beta_1^l,\beta_2^l)\sqrt{1-{|v|^2\over c^2}}
\over \alpha|v|^2(1+{|x|\over \sqrt{2}}-|q|)^{\alpha}}\le 1. \label{5.1a}
\end{equation}
Then there exists a unique solution $z_\pm(w,x+q,.)$ of the equation \eqref{5.3}
so that
\begin{equation}
\dot z_\pm(w,x+q,t)-w=o(1)\textrm{ as }t\to \pm\infty,\  z_\pm(w,x+q,0)=x+q,\label{5.3a}
\end{equation}
and
\begin{equation}
\sup_\R |\dot z_\pm(w,x+q,.)-w|
\le {2^{9\over 2}n^{3\over 2}\beta_1^l\sqrt{1-{|v|^2\over c^2}}\over \alpha |v|(1+{|x|\over \sqrt{2}}-|q|)^\alpha}.\label{5.3b}
\end{equation}
\end{lemma}
A proof of Lemma \ref{lem_scatinit} is given in Section \ref{sec_scatinit}.
For the rest of this Section we set
\begin{eqnarray}
\mu^l&:=&\mu\big(2^8\alpha^{-1}n^2\max(\beta_1^l,\beta_2^l)\big),\label{5.3dd}\\
z_\pm(v,t)&:=&z_\pm(v,0,t)\textrm{ for }t\in \R,\textrm{ when }|v| \ge \mu^l,\label{5.3ee}\\
\beta_2&:=&\max(\beta_2^l,\beta_2^s),\label{5.3f}
\end{eqnarray} 
where the function $\mu$ is defined by \eqref{1.2b}. 

For the rest of the text  $H(f(\tau),\dot f(\tau))$ is shortened to $H(f)(\tau)$ for any $(f,\tau)\in C^1(\R,\R^n)\times\R$, where $H$ stands for $F$, $F^s$ or $F^l$.

Let $(v_-,x_-)\in \B(0,c)\times \R^n$, $v_-\cdot x_-=0$ and $|v_-|\ge \mu^l$.
Then the function $y_-$ in \eqref{1.6} satisfies the integral equation
$y_-=A(y_-)$
where
\begin{eqnarray}
A(f)(t)&=&\int_{-\infty}^t \dot A(f)(\sigma)d\sigma,\label{2.1a}\\
\dot A(f)(t)&=&g\big(g^{-1}(v_-)+\int_{-\infty}^t F(z_-(v_-,.)+x_-+f)(\tau)d\tau\big)\nonumber\\
&&-g\big(g^{-1}(v_-)+\int_{-\infty}^t F^l(z_-(v_-,.))(\tau)d\tau\big),\label{2.6}
\end{eqnarray}
for $t\in \R$ and for $f\in C^1(\R,\R^n)$, $\sup_{(-\infty,0]}(|f|+|\dot f|)<\infty$.
We have $A(f)\in C^2(\R,\R^n)$ for $f\in C^1(\R,\R^n)$ so that $\sup_{(-\infty,0]}(|f|+|\dot f|)<\infty$ (see \eqref{02b}, \eqref{03b} and \eqref{03c}). 

For $r\in (0,1)$ and for $|v_-|\ge \mu^l$, $|v_-|\ge 2^{3\over 2}r$, we introduce the following complete metric space $M_{r,v_-}$ endowed with the following norm $\|.\|$
\begin{equation}
M_{r,v_-}=
\lbrace
f\in C^1(\R,\R^n)\ |\ \sup_{\R}|\dot z_-(v_-,.)+\dot f|\le c,\ \|f\|\le r\},\label{2.3}
\end{equation}
\begin{equation}
\|f\|=\max\big(\sup_{t\in (-\infty,0)}\max\big(1,\big(1-r+({|v_-|\over 2\sqrt{2}}-r)|t|\big)\big)|\dot f(t)|,
\sup_{(0,+\infty)}|\dot f|,\sup_{(-\infty, 0)}|f|\Big).\label{2.3b}
\end{equation}
The space $M_{r,v_-}$ is a convex subset of $C^1(\R,\R^n)$.
Then we have the following estimate and contraction estimate for the map $A$ restricted to $M_{r,v_-}$.

\begin{lemma}
\label{lem_cont}
Let $(v_-,x_-)\in \big(\B(0,c)\b\B(0,\mu^l)\big)\times\R^n$, $v_-\cdot x_-=0$,  and let $r\in\big(0,\min({|v_-|\over 2^{3\over 2}},1)\big)$.
When
\begin{equation}
{2^{5\over 2}n\max(\beta_1^l,\beta_2)\sqrt{1-{|v_-|^2\over c^2}}\over \alpha\big({|v_-|\over 2\sqrt{2}}-r\big)^2\big(1-r\big)^{\alpha+1}}\le 1,\label{lc0}
\end{equation}
then the following estimates are valid
\begin{eqnarray}
\|A(f)\|&\le&\lambda_1(n,\alpha,\beta_1^l,\beta_2,|x_-|,|v_-|,r)\label{lc1}\\
&:=&{4n^{3\over 2}\sqrt{1-{|v_-|^2\over c^2}}\big({r\beta_1^l\over c}+2\beta_2(n^{1\over 2}(|x_-|+r)+1)\big)\big(1+{1\over {|v_-|\over 2\sqrt{2}}-r}\big)\over \alpha({|v_-|\over 2\sqrt{2}}-r)(1-r)^{\alpha+1}}
,\nonumber
\end{eqnarray}
and
\begin{equation}
\|A(f_1)-A(f_2)\|\le\lambda_2(n,\alpha,\beta_2,\beta_3^s,|v_-|,r)\|f_1-f_2\|,\label{lc2}
\end{equation}
\begin{equation*}
\lambda_2(n,\alpha,\beta_2,\beta_3^s,|v_-|,r):={4n^{3\over 2}\sqrt{1-{|v_-|^2\over c^2}}\big({\beta_1^l+\beta_2\over c}+2n^{1\over 2}(\beta_2+\beta_3^s)\big)\big(1+{1\over {|v_-|\over 2\sqrt{2}}-r}\big)\over \alpha({|v_-|\over 2\sqrt{2}}-r)(1-r)^{\alpha+2}}
,
\end{equation*}
for $(f,f_1,f_2)\in M_{r,v_-}^3$.
\end{lemma}
A proof of Lemma \ref{lem_cont} is given in Section \ref{sec_cont}.

We also need the following result.

\begin{lemma}
\label{lem_decomp}
Let $(v_-,x_-)\in (\B(0,c)\b \B(0,\mu^l))\times\R^n$, $v_-\cdot x_-=0$, and let $r\in\big(0,\min({|v_-|\over 2^{3\over 2}},1)\big)$.
When $y_-\in M_{r,v_-}$ is a fixed point for the map $A$ then $x:=z_-(v_-,.)+x_-+y_-$ is a scattering solution for equation \eqref{1.1} and
\begin{equation}
x(t)=z_+(a(v_-,x_-),t)+b(v_-,x_-)+y_+(t),\label{3.13}
\end{equation} 
for $t\ge0$, where
\begin{equation}
a(v_-,x_-):=g\Big(g^{-1}(v_-)+\int_{-\infty}^{+\infty}F(x)(\tau)d\tau\Big),\label{300a}
\end{equation}
\begin{equation}
b(v_-,x_-):=x_-+A(y_-)(0)-y_+(0),\label{3.13a}
\end{equation}
\begin{eqnarray}
y_+(t)&=&-\int_t^{+\infty}\Big(g\big(g^{-1}(a(v_-,x_-))-\int_\sigma^{+\infty} F(x)(\tau)d\tau\big)
\label{3.13b}\\
&&-g\big(g^{-1}(a(v_-,x_-))-\int_\sigma^{+\infty} F^l(z_+(a(v_-,x_-),.))(\tau)d\tau\big)\Big)d\sigma,\nonumber
\end{eqnarray}
for $t\ge 0$. 
\end{lemma}
Lemma \ref{lem_decomp} is proved in Section \ref{sec_scatinit}.

\subsection{Estimates on the scattering solutions}
In this Section our main results consist in estimates and asymptotics for the scattering data $(a_{sc}, b_{sc})$ and scattering solutions for the equation \eqref{1.1}.
\begin{theorem}
\label{thm_y}
Under the assumptions of Lemma \ref{lem_decomp} and when 
\begin{equation}
{24n^2\sqrt{1-{|v_-|^2\over c^2}}\max(\beta_1^l,\beta_2)(1+{1\over c})(1+{1\over {|v_-|\over 2^{3\over 2}}-r})\over \alpha({|v_-|\over 2^{3\over 2}}-r)(1-r)^{\alpha+1}}\le 1,\label{t1h}
\end{equation}
then the following estimates are valid:
\begin{equation}
|\dot y_-(t)|\le 
{2n^{3\over 2}\left(1-{|v_-|^2\over c^2}\right)^{1\over 2}\big({r\beta_1^l\over c}+2\beta_2(n^{1\over 2}(|x_-|+r)+1)\big)\over (\alpha+1)\big({|v_-|\over 2\sqrt{2}}-r\big)\big(1-r-\big({|v_-|\over 2\sqrt{2}}-r\big)t\big)^{\alpha+1}},\label{t1a}
\end{equation}
for $t\le 0$. In addition
\begin{equation}
|a_{sc}(v_-,x_-)|\le {8n^{3\over 2}\sqrt{1- {|v_-|^2\over c^2}}\over ({|v_-|\over 2^{3\over 2}}-r)(1+{|x_-|\over \sqrt{2}}-r)^\alpha}\big({\beta_1^l\over \alpha}
+{\beta_2\over (\alpha+1)(1+{|x_-|\over \sqrt{2}}-r)}\big),\label{t1c}
\end{equation}
\begin{equation}
|b_{sc}(v_-,x_-)|\le{4n^{3\over 2}\sqrt{1-{|v_-|^2\over c^2}}\big({r\beta_1^l\over c}+\beta_2(n^{1\over 2}(4|x_-|+r)+4)\big)
\over \alpha (\alpha+1)({|v_-|\over 2\sqrt{2}}-r)^2(1-r)^\alpha},\label{t1d}
\end{equation}
\begin{equation}
|\dot y_+(t)|\le{2n^{3\over 2}\big({r\beta_1^l\over c}+2\beta_2(n^{1\over 2}(3|x_-|+r)+3)\big)\sqrt{1-{|v_-|^2\over c^2}}\over (\alpha+1)\big({|v_-|\over 2\sqrt{2}}-r\big)\big(1-r+t({|v_-|\over 2^{3\over 2}}-r)\big)^{\alpha+1}},\label{t1e}
\end{equation}
for $t\ge 0$, and 
\begin{eqnarray}
&&|a_{sc}(v_-,x_-)-\sqrt{1-{|v_-|^2\over c^2}}\int_{-\infty}^{+\infty}F(\tau v_-+x_-,v_-)d\tau|\nonumber\\
&\le& {520n^4\beta^2\big(1-{|v_-|^2\over c^2}\big)
\big({r\over c}+|x_-|+1\big)({1\over c}+1)\big(1+{1\over {|v_-|\over 2\sqrt{2}}-r}\big)^2\over \alpha^2({|v_-|\over 2\sqrt{2}}-r)^2(1-r)^{2\alpha+3}}
,\label{t3}
\end{eqnarray}
\begin{eqnarray}
&&\Big|b_{sc}(v_-,x_-)-W(v_-,x_-)-\sqrt{1-{|v_-|^2\over c^2}}\Big(\int_{-\infty}^0\int_{-\infty}^\sigma F^s(\tau v_-+x_-,v_-)d\tau d\sigma\nonumber\\
&&-\int_0^{+\infty}\int_\sigma^{+\infty} F^s(\tau v_-+x_-,v_-)d\tau d\sigma+\int_{-\infty}^{+\infty}V^s(\sigma v_-+x_-)d\sigma {v_-\over c^2}\Big)\Big|\nonumber\\
&\le&{468n^4\beta^2\big(1-{|v_-|^2\over c^2}\big)\big({r\over c}+|x_-|+1\big)({1\over c}+1)\big(1+{1\over  
{|v_-|\over 2\sqrt{2}}-r}\big)^2\over \alpha^2(\alpha+1)({|v_-|\over 2\sqrt{2}}-r)^3(1-r)^{2\alpha+2}},\label{t4}
\end{eqnarray}
where $\beta=\max(\beta_1^l,\beta_2,\beta_3^s)$.
\end{theorem}
Theorem \ref{thm_y} is proved in Section \ref{sec_thm_y}.

\begin{proof}[Proof of Theorem \ref{thm}]
Let $(\theta,x)\in T\S^{n-1}$ and let $(r,\rho)\in (0,\min(1,2^{-{3\over 2}}c))\times (0,+\infty)$, $\rho>\rho_0(|x|,r,\beta,\alpha)$, where $\rho_0$ is defined in \eqref{1.30}. Set $(v_-,x_-)=(\rho\theta,x)$. Then 
note that
\begin{equation}
\max\big(\lambda_0,{\lambda_1\over r},\lambda_2,\lambda_3\big)
\le{32n^2\sqrt{1-{|v_-|^2\over c^2}}\beta(1+|x_-|+{1\over c})(1+{1\over {|v_-|\over 2^{3\over 2}}-r})\over \alpha({|v_-|\over 2^{3\over 2}}-r)r(1-r)^{\alpha+2}}<1,\label{801}
\end{equation}
where $\lambda_1$ and 
$\lambda_2$ are defined in \eqref{lc1} and \eqref{lc2} respectively, and where 
$\lambda_0:=
{2^8n^2\max(\beta_1^l,\beta_2)\sqrt{1-{|v_-|^2\over c^2}}
\over \alpha|v_-|^2}$ and $\lambda_3$ is the left-hand side of \eqref{t1h} (we also used \eqref{1.30}).
From estimate \eqref{801} and Lemma \ref{lem_cont} we obtain:  $|v_-|\ge \mu^l$ (see \eqref{5.3dd}), $A$ is a contraction in $M_{r,v_-}$, and Theorem \ref{thm_y} holds for the unique fixed point $y_-\in M_{r,v_-}$ of $A$. The estimate \eqref{t3} and \eqref{t4} hold and they provide the estimates \eqref{t3b} and \eqref{t4b}, which proves Theorem \ref{thm}.
\end{proof}

\subsection{Motivations for changing the definition of the scattering map}
\label{com}
For a solution  $x$ at a nonzero energy for equation \eqref{1.1} we say that it is a scattering solution when there exists $\ep>0$ so that $1+|x(t)|\ge \ep(1+|t|)$ for $t\in \R$ (see \cite{DG}).
In the Introduction and in the previous subsections we choose to parametrize the scattering solutions of equation \eqref{1.1} by the solutions $z_\pm(v,.)$ of the equation \eqref{5.3} (see the asymptotic behaviors \eqref{1.6} and \eqref{1.7}), and then to formulate the inverse scattering problem \eqref{P} using this parametrization. 
We obtain the estimates \eqref{t3b} and \eqref{t4b} that provide the high energies asymptotics and the Born approximation at fixed energy of the scattering data. However these estimates 
do not provide the asymptotics of the scattering data
$(a_{sc},b_{sc})$ when the parameters $\alpha,$ $n,$ $\rho$ and $\beta$ are fixed and $|x|\to +\infty$.
Motivated by this disadvantage, in the next section we modify the definition of the scattering map given in the Introduction so that one can obtain a result on this asymptotic regime. 

\section{A modified scattering map}
\subsection{Changing the parametrization of the scattering solutions}
We set
\begin{equation}
\mu^l_\sigma:=\mu\big({2^8\alpha^{-1}n^2\max(\beta_1^l,\beta_2^l)(1+{\sigma\over \sqrt{2}})^{-\alpha}}\big),\label{0.0}
\end{equation}
for $\sigma\ge 0$, where the function $\mu$ is defined by \eqref{1.2b}.

Under conditions \eqref{1.4a} and \eqref{1.4b}, the following is valid: for any 
$(v_-,x_-)\in \B(0,c)\times\R^n$ so that $|v_-|> \mu^l_{|x_-|}$, then
the equation \eqref{1.1}  has a unique solution $x\in C^2(\R,\R^n)$ such that
\begin{equation}
{x(t)=z_-(v_-,x_-,t)+y_-(t),}\label{4.6}
\end{equation}
where $\dot y_-(t)\to 0,\ y_-(t)\to 0,\ {\rm as}\ t\to -\infty;$

In addition the function $y_-$ in \eqref{4.6} satisfies the integral equation
$y_-=\A(y_-)$
where
\begin{eqnarray}
\A(f)(t)&=&\int_{-\infty}^t\dot\A(f)(\sigma)d\sigma,\label{4.1a}\\
\dot\A(f)(t)&=&g\big(g^{-1}(v_-)+\int_{-\infty}^t F(z_-(v_-,x_-,.)+f)(\tau)d\tau\big)\nonumber\\
&&-g\big(g^{-1}(v_-)+\int_{-\infty}^t F^l(z_-(v_-,x_-,.))(\tau)d\tau\big),\label{4.7}
\end{eqnarray}
for $t\in \R$ and for $f\in C^1(\R,\R^n)$, $\sup_{(-\infty,0]}(|f|+|\dot f|)<\infty$.
We remind that  $H(f(\tau),\dot f(\tau))$ is shortened to $H(f)(\tau)$ for any $(f,\tau)\in C^1(\R,\R^n)\times\R$ above and in the rest of the text, where $H$ stands for $F$, $F^s$ or $F^l$. 

We have $\A(f)\in C^2(\R,\R^n)$ for $f\in C^1(\R,\R^n)$ so that $\sup_{(-\infty,0]}(|f|+|\dot f|)<\infty$ (see \eqref{02b}, \eqref{03b} and \eqref{03c}). 

For $r\in\big(0,\min\big(1,{|v_-|\over 2^{3\over 2}}\big)\big)$, we introduce the following metric space $M_{r,v_-,x_-}$ endowed with the following norm $\|.\|_*$
\begin{equation}
M_{r,v_-,x_-}=
\lbrace
f\in C^1(\R,\R^n)\ |\ \sup_{\R}|\dot z_-(v_-,x_-,.)+\dot f|\le c,\ \|f\|_*\le r
\rbrace,\label{4.3}
\end{equation}
\begin{equation}
\|f\|_*=\max\big(\sup_{t\in (-\infty,0)}\max\big(1,\big(1-r+{|x_-|\over \sqrt{2}}+({|v|\over 2^{3\over 2}}-r)|t|\big)\big)
|\dot f(t)|, \sup_{(0,+\infty)}|\dot f|,\sup_{(-\infty,0)}|f|\Big).\label{4.3b}
\end{equation}
The space $M_{r,v_-,x_-}$ is a convex subset of $C^1(\R,\R^n)$.
We study the map $\A$ defined by \eqref{4.1a} and \eqref{4.7} on the metric space $M_{r,v_-,x_-}$.
Set
\begin{equation}
\tilde k(v_-,x_-,f):=g\big(g^{-1}(v_-)+\int_{-\infty}^{+\infty}F(z_-(v_-,x_-,.)+f)(t)dt\big),\label{5.18aa}
\end{equation}
for $f\in M_{r,v_-,x_-}$.
For the rest of the section we also set $\beta_2=\max(\beta_2^l,\beta_2^s)$.

The following Lemma \ref{lem_cont2} is the analog of Lemma \ref{lem_cont}.

\begin{lemma}
\label{lem_cont2}
Let $(v_-,x_-)\in \B(0,c)\times\R^n$, $v_-\cdot x_-=0$, $|v_-|>\mu^l_{|x_-|}$,  and let $r\in \big(0,\min(1,{|v_-|\over 2^{3\over 2}})\big)$.
When
\begin{equation}
{2^{3\over 2}n\max(\beta_1^l,\beta_2)\sqrt{1-{|v_-|^2\over c^2}}
\big(1+{1\over 1+{|x_-|\over \sqrt{2}}-r}\big)\over \alpha\big({|v_-|\over 2\sqrt{2}}-r\big)^2\big(1+{|x_-|\over \sqrt{2}}-r\big)^\alpha}\le 1,\label{lc2_0}
\end{equation}
then the following estimates are valid
\begin{equation}
\|\A(f)\|_*\le\tilde \lambda_1(n,\alpha,\beta_1^l,\beta_2,|x_-|,|v_-|,r),\label{lc3}
\end{equation}
\begin{equation*}
\tilde \lambda_1:={2n^{3\over 2}\big(1-{|v_-|^2\over c^2}\big)^{1\over 2}({\beta_1^l r\over c}+4\beta_2(n^{1\over 2}r+1))\over \alpha({|v_-|\over 2\sqrt{2}}-r)(1-r+{|x_-|\over\sqrt{2}})^\alpha}\big(1+{1\over{|v_-|\over 2\sqrt{2}}-r}+{1\over 1-r+{|x_-|\over\sqrt{2}}}\big),
\end{equation*}
and
\begin{equation}
\|\A(f_1)-\A(f_2)\|_*\le\tilde\lambda_2(n,\alpha,\beta_2,\beta_3^s,|x_-|,|v_-|,r)\|f_1-f_2\|_*,\label{lc4}
\end{equation} 
\begin{equation*}
\tilde\lambda_2:={4n^{3\over 2}\big(1-{|v_-|^2\over c^2}\big)^{1\over 2}({\beta_1^l \over c}+2\beta_2n^{1\over 2}+{{\beta_2 \over c}+2\beta_3^sn^{1\over 2}\over 1-r+{|x_-|\over\sqrt{2}}})\over \alpha({|v_-|\over 2\sqrt{2}}-r)(1-r+{|x_-|\over\sqrt{2}})^\alpha}\big(1+{1\over{|v_-|\over 2\sqrt{2}}-r}+{1\over 1-r+{|x_-|\over\sqrt{2}}}\big).
\end{equation*}
for $(f,f_1,f_2)\in M_{r,v_-,x_-}^3$.
In addition we have
\begin{equation}
|\tilde k(v_-,x_-,f)-v_-|\le {8n^{3\over 2}\sqrt{1-{|v_-|^2\over c^2}}\over ({|v_-|\over 2\sqrt{2}}-r) (1+{|x_-|\over \sqrt{2}}-r)^\alpha}\big({\beta_1^l\over \alpha}+{\beta_2\over (\alpha+1)(1+{|x_-|\over \sqrt{2}}-r)}\big),\label{lc7}
\end{equation}
for $f\in M_{r,v_-,x_-}$.
\end{lemma}
Lemma \ref{lem_cont2} is proved in Section \ref{sec_cont}.

Let $(v_-,x_-)\in \B(0,c)\times\R^n$, $v_-\cdot x_-=0$  and let $r\in\big(0,\min({1\over 2},{|v_-|\over 2^{3\over 2}})\big)$. 
Under the following condition  
\begin{equation}
{48n^2(c^{-1}+1)\max(\beta_1^l,\beta_2)\sqrt{1-{|v_-|^2\over c^2}}\over \alpha({|v_-|\over 2^{3\over 2}}-r)({1\over 2}+{|x_-|\over\sqrt{2}})^\alpha}
(1+{1\over {|v_-|\over 2^{3\over 2}}-r})\le 1,\label{10.1}
\end{equation}
then condition \eqref{5.1a} is satisfied for any $q\in\overline{\B(0,{1\over 2})}$, and condition \eqref{lc2_0} is also satisfied. In particular, $|v_-|\ge \mu_{|x_-|}^l$, and 
when $y_-\in M_{r,v_-,x_-}$ is a fixed point of the operator $\A$ then $x:=z_-(v_-,x_-,.)+y_-$ is a scattering solution of \eqref{1.1} (in the sense given in Section \ref{com}). We set
\begin{equation}
\tilde a(v_-,x_-):=\tilde k(v_-,x_-,y_-).\label{5.18a}
\end{equation} 
By conservation of energy $|\tilde a(v_-,x_-)|=\lim_{t\to\infty}|\dot x(t)|=|v_-|$. 
From \eqref{lc7} and condition \eqref{10.1} it follows that $|\tilde a(v_-,x_-)-v_-|\le2^{-{5\over 2}}|v_-|$, and 
we can consider the free solution $z_+(\tilde a(v_-,x_-),x_-+q,.)$ for any $q\in \overline{\B(0,{1\over 2})}$ (see Lemma \ref{lem_scatinit}). Furthermore, with appropriate changes in the proof of Lemma \ref{lem_decomp}, one can prove that the following decomposition holds 
\begin{equation}
x(t)=z_+(\tilde a(v_-,x_-),x_-+q,t)+\G_{v_-,x_-}(q)-q+h(v_-,x_-,q,t),\label{6.6}
\end{equation}
where      
\begin{equation}
\G_{v_-,x_-}(q):=\A(y_-)(0)-h(v_-,x_-,q,0),\label{6.1b}
\end{equation}
and
\begin{eqnarray}
h(v_-,x_-,q,t)
&:=&-\int_t^{+\infty}\Big(g\big(g^{-1}(\tilde a(v_-,x_-))-\int_\sigma^{+\infty}F(x)(\tau))d\tau\big)\label{6.6c}\\
&&\hskip-2cm -g\big(g^{-1}(\tilde a(v_-,x_-))-\int_\sigma^{+\infty}F^l(z_+(\tilde a(v_-,x_-),x_-+q,.))(\tau)d\tau\big)\Big)d\sigma,\nonumber
\end{eqnarray}
for $t\ge 0$ and for $q\in \overline{\B(0,{1\over 2})}$. We need the following Lemma.

\begin{lemma}
\label{lem_cont3}
Let $(v_-,x_-)\in \B(0,c)\times\R^n$, $v_-\cdot x_-=0$, and let $r\in\big(0,\min({1\over 2},{|v_-|\over 2^{3\over 2}})\big)$. Under condition \eqref{10.1} and when $y_-\in M_{r,v_-,x_-}$ is a fixed point of $\A$, then
\begin{equation}
|\G_{v_-,x_-}(q)|\le{4n^{3\over 2}\sqrt{1-{|v_-|^2\over c^2}}({\beta_1^l\over c}+2n^{1\over 2}\beta_2+4\beta_2)\over \alpha(\alpha+1)\big({|v_-|\over 2^{3\over 2}}-r\big)^2\big({1\over 2}+{|x_-|\over \sqrt{2}}\big)^\alpha}
\le{1\over 2},\label{lc3a}
\end{equation}
for $|q|\le {1\over 2}$, and
\begin{equation}
|\G_{v_-,x_-}(q)-\G_{v_-,x_-}(q')|\le{4n^{3\over 2}\sqrt{1-{|v_-|^2\over c^2}}({\beta_1^l\over c}+2n^{1\over 2}\beta_2^l)|q-q'|\over \alpha(\alpha+1)({|v_-|\over 2^{3\over 2}})^2({1\over 2}+{|x_-|\over \sqrt{2}})^\alpha}\le {|q-q'|\over 6},\label{lc3b}
\end{equation}
for $(q,q')\in \overline{\B(0,{1\over 2})}^2$.
\end{lemma}
Lemma \ref{lem_cont3} is proved in Section \ref{sec_cont3}.

Under the assumptions of Lemma \ref{lem_cont3} the map $\G_{v_-,x_-}$ is a ${1\over 6}$-contraction map from $\overline{\B(0,{1\over 2})}$ to $\overline{\B(0,{1\over 2})}$. We denote by $\tilde b_{sc}(v_-,x_-)$ its unique fixed point in $\overline{\B(0,{1\over 2})}$, and we set $\tilde b(v_-,x_-):=x_-+\tilde b_{sc}(v_-,x_-)$ and $\tilde a_{sc}(v_-,x_-):= \tilde a(v_-,x_-)-v_-$. The decomposition \eqref{6.6} becomes
\begin{eqnarray}
&&z_-(v_-,x_-,t)+y_-(t)
=z_+(\tilde a(v_-,x_-),\tilde b(v_-,x_-),t)+y_+(t),\label{6.6a}\\
&&y_+(t)=h(v_-,x_-,\tilde b_{sc}(v_-,x_-),t),\label{6.6b}
\end{eqnarray}
for $t\ge 0$. The map $(\tilde a_{sc}, \tilde b_{sc})$ are our modified scattering data. The inverse scattering problem for equation \eqref{1.1} can now be formulated as follows 
\begin{equation}
\textrm{Given }(\tilde a_{sc},\tilde b_{sc})\textrm{ and }F^l,\textrm{ find }F^s.\label{P2}
\end{equation}

\subsection{Estimates and asymptotics of the modified scattering data}
Let $r\in\big(0,\min({1\over 2},2^{-{3\over 2}}c)\big)$ and let $(v_-,x_-)\in \B(0,c)\times\R^n$, $v_- \cdot x_-=0$ so that $|v_-|>\tilde \rho_0(|x_-|,r,\beta,\alpha)$, where $\beta=\max(\beta_1^l,\beta_2,\beta_3^s)$ and $\tilde \rho_0$ is defined as the root of the following equation
\begin{equation}
1={72n^2(c^{-1}+1)\beta\sqrt{1-{\tilde\rho_0^2\over c^2}}(1+{1\over {\tilde\rho_0\over 2^{3\over 2}}-r})\over \alpha r({\tilde\rho_0\over 2^{3\over 2}}-r)({1\over 2}+{|x_-|\over\sqrt{2}})^\alpha}
,\ \tilde \rho_0\in(2^{3\over 2}r,c).\label{17.30}
\end{equation}
Then note that
\begin{equation}
\max\big({\tilde\lambda_1\over r},\tilde\lambda_2,\tilde \lambda_3\big)
\le{72n^2(c^{-1}+1)\beta\sqrt{1-{\rho^2\over c^2}}(1+{1\over {\rho\over 2^{3\over 2}}-r})\over \alpha r({\rho\over 2^{3\over 2}}-r)({1\over 2}+{|x_-|\over\sqrt{2}})^\alpha},\ \rho \in(2^{3\over 2}r,c),\label{802}
\end{equation}
where $\tilde \lambda_1(n,\alpha,\beta_1^l,\beta_2,|x_-|,\rho,r)$ and $\tilde \lambda_2(n,\alpha,\beta_2,\beta_3^s,|x_-|,\rho,r)$ are defined in \eqref{lc3} and \eqref{lc4} respectively, and where 
$\tilde \lambda_3$ is the left-hand side of \eqref{10.1} (for $|v_-|$ replaced by $\rho$ in \eqref{10.1}).
From estimate \eqref{802}, \eqref{17.30} and Lemma \ref{lem_cont2} we obtain:  $|v_-|\ge \mu^l_{|x_-|}$ (see \eqref{0.0}), $\A$ is a contraction in $M_{r,v_-,x_-}$, and $\A$ has a unique fixed point $y_-\in M_{r,v_-,x_-}$ of $\A$. 

Then set
\begin{eqnarray}
\tilde W(v_-,x_-)&=&g\big(g^{-1}(v_-)+\int_{-\infty}^0 F^l(z_-(v_-,x_-,.))(\tau)d\tau\label{t2.0}\\
&&+\int_0^{+\infty}F^l(z_+(\tilde a(v_-,x_-),x_-,.))(\tau)d\tau\big)-v_-.\nonumber
\end{eqnarray}
Note that $\tilde W$ is known from the modified scattering data and from $F^l$.
We obtain the following analog of Theorem \ref{thm_y}.
  
\begin{theorem}
\label{thm_y2}
Let $r\in\big(0,\min({1\over 2},2^{-{3\over 2}}c)\big)$ and let $(v_-,x_-)\in \B(0,c)\times\R^n$, $v_- \cdot x_-=0$ so that $|v_-|\ge \tilde \rho_0(|x_-|,r,\beta,\alpha)$ where $\beta=\max(\beta_1^l,\beta_2,\beta_3^s)$ and $\tilde\rho_0$ is defined by \eqref{17.30}. 
Under conditions \eqref{1.4a}, \eqref{1.4b} the following estimates are valid:
\begin{eqnarray}
|\dot y_-(t)|&\le&
{2n^{3\over 2}\left({1-{|v_-|^2\over c^2}}\right)^{1\over 2}\big({r\beta_1^l\over c}+2\beta_2(n^{1\over 2}r+1)\big)\over (\alpha+1)\big({|v_-|\over 2\sqrt{2}}-r\big)\big(1-r+{|x_-|\over \sqrt{2}}-\big({|v_-|\over 2\sqrt{2}}-r\big)t\big)^{\alpha+1}},\label{t2a}
\end{eqnarray}
for $t \le 0$; and 
\begin{equation}
|\tilde a_{sc}(v_-,x_-)|\le {24n^{3\over 2}\left({1-{|v_-|^2\over c^2}}\right)^{1\over 2}\max(\beta_1^l,\beta_2)\over \alpha\big({|v_-|\over 2\sqrt{2}}-r\big)\big({1\over 2}+{|x_-|\over \sqrt{2}}\big)^\alpha},\label{t2c}
\end{equation}
\begin{equation}
|\tilde b_{sc}(v_-,x_-)|\le{24n^2\left({1-{|v_-|^2\over c^2}}\right)^{1\over 2}\max(\beta_1^l,\beta_2)({1\over c}+1)
\over \alpha(\alpha+1)\big({|v_-|\over 2\sqrt{2}}-r\big)^2\big({1\over 2}+{|x_-|\over\sqrt{2}}\big)^\alpha},\label{t2d}
\end{equation}
\begin{equation}
|\dot y_+(t)|\le{20n^4\left({1-{|v_-|^2\over c^2}}\right)^{1\over 2}\max(\beta_1^l,\beta_2)({1\over c}+1)
\over (\alpha+1)\big({|v_-|\over 2\sqrt{2}}-r\big)\big({1\over 2}+{|x_-|\over \sqrt{2}}+t{|v_-|\over 2\sqrt{2}}\big)^{\alpha+1}}.\label{t2e}
\end{equation}
for $t\ge 0$. In addition
\begin{eqnarray}
&&|\tilde a_{sc}(v_-,x_-)-\tilde W(v_-,x_-)-\sqrt{1-{|v_-|^2\over c^2}}\int_{-\infty}^{+\infty}F^s(\tau v_-+x_-,v_-)d\tau|\nonumber\\
&\le&{944n^4\beta^2\big(1-{|v_-|^2\over c^2}\big)({r\over c}+1)(1+{1\over c})(1+{1\over ({|v_-|\over 2^{3\over 2}}-r)})^2\over \alpha^2({|v_-|\over 2^{3\over 2}}-r)^2
({1\over 2}+{|x_-|\over \sqrt{2}})^{2\alpha+1}}.\label{t2f}
\end{eqnarray}
\begin{eqnarray}
&&\Big|\tilde b_{sc}(v_-,x_-)-\sqrt{1-{|v_-|^2\over c^2}}\Big(\int_{-\infty}^0\int_{-\infty}^\sigma F^s(\tau v_-+x_-,v_-)d\tau d\sigma\nonumber\\
&&-\int_0^{+\infty}\int_\sigma^{+\infty} F^s(\tau v_-+x_-,v_-)d\tau d\sigma+\int_{-\infty}^{+\infty}V^s(\tau v_-+x_-)d\tau {v_-\over c^2}\Big)\Big|\nonumber\\
&\le&{808n^4\beta^2\big(1-{|v_-|^2\over c^2}\big)({1\over c}+1)^2\big(1+{1\over {|v_-|\over 2\sqrt{2}}-r}\big)^2
\over \alpha^2(\alpha+1)\big({|v_-|\over 2\sqrt{2}}-r\big)^3\big({1\over 2}+{|x_-|\over\sqrt{2}}\big)^{2\alpha}}.
\label{t2g}
\end{eqnarray}
\end{theorem}
Theorem \ref{thm_y2} is proved in Section \ref{sec_cont3}.

Estimates \eqref{t2f} and \eqref{t2g} provide the high energies asymptotics of the modified scattering data.
The analog of formulas \eqref{t1} and \eqref{t2} are given in the following Corollary.

\begin{corollary}
\label{high_energy} 
Let $(\theta,x)\in TS^{n-1}$. We have 
\begin{eqnarray}
\lim_{\rho\to +\infty} {\rho\over \sqrt{1-{\rho^2\over c^2}}} (\tilde a_{sc}-\tilde W)(\rho\theta,x)&=&\int_{-\infty}^{+\infty}F(\tau \theta+x,c\theta)d\tau,\label{t1m}\\
\lim_{s\to +\infty}{\rho^2\over \sqrt{1-{\rho^2\over c^2}}}\tilde b_{sc}(s\theta,x)&=&\int_{-\infty}^0\int_{-\infty}^\sigma F^s(\tau \theta+x,c\theta)d\tau d\sigma
\nonumber\\
&&\hskip -4cm-\int_0^{+\infty}\int_\sigma^{+\infty} F^s(\tau \theta+x,c\theta)d\tau d\sigma+PV^s(\theta,x)\theta.\label{t2m}
\end{eqnarray}
\end{corollary}
Then $F^s$ can be reconstructed from the high energies asymptotics of $\tilde a_{sc}$, 
and from \eqref{t2m} one can prove the following statements:
The potential $V^s$ is uniquely determined up to its radial part by $\tilde b_{sc}$; The magnetic field $B^s$ can be reconstructed from $\tilde b_{sc}$ when $n\ge 3$, and up to its radial part when $n=2$.

The estimates \eqref{t2f} and \eqref{t2g} also gives the Born approximation for the modified scattering data at fixed energy when
the electromagnetic field is sufficiently weak.
The Born approximation at fixed energy of $\tilde a_{sc}(\rho\theta,x)-\tilde W(\rho\theta,x)$ and
$\tilde b_{sc}(\rho\theta,x)$
for $(\theta,x)\in T\S^{n-1}$ and for $\tilde \rho\in (\tilde \rho_0(|x|,r,\beta,\alpha),c)$ are given by \eqref{t3c} and \eqref{t4c}
respectively, and we have:
The force $F^s$ can be reconstructed from the Born approximation \eqref{t3c} of $\tilde a_{sc}$ at fixed energy;
$V^s$ can be reconstructed from the Born approximation \eqref{t4c} of $\tilde b_{sc}$ at fixed energy; 
$B^s$ can be reconstructed from \eqref{t4c} when $n\ge 3$, and up to its radial part when $n=2$.

Estimates \eqref{t2f} and \eqref{t2g} also provide the first leading term in the asymptotics of the modified scattering data $(\tilde a_{sc}(\rho\theta,x),\tilde b_{sc}(\rho\theta,x))$ when the parameters $\alpha$, $n$, $s$, $\theta$ and $\beta$ are fixed and $|x|$ increases to $+\infty$.

\section{Preliminary estimates and proof of Lemmas \ref{lem_scatinit} and \ref{lem_decomp}}
\label{sec_scatinit}
For the rest of the text we use the following properties of the function $g:\R^n\to \B(0,c)$ defined by \eqref{01} (see \cite[Section 2.5]{JoT}):
\begin{eqnarray}
|\nabla g_i(x)|^2&\le&{1\over 1+ {|x|^2\over c^2}},\label{02a}\\
|g(x)-g(y)|&\le & \sqrt{n}|x-y|\sup_{\ep\in (0,1)}{1\over \sqrt{1+{|\ep x+(1-\ep)y|^2\over c^2}}},\label{02b}
\end{eqnarray}
\begin{eqnarray}
|\nabla g_i(x)-\nabla g_i(y)|&\le &{3\sqrt{n}\over c}|x-y|\sup_{\ep\in (0,1)}{1\over 1+{|\ep x+(1-\ep)y|^2\over c^2}},\label{02c}\\
\nabla g_i(x)&=&{1\over (1+{|x|^2\over c^2})^{1\over 2}}e_i-{x_ix\over c^2(1+{|x|^2\over c^2})^{3\over 2}},\label{02d}
\end{eqnarray}
for $(x,y)\in \R^n\times \R^n$, $x=(x_1,\ldots, x_n)$, and for $i=1\ldots n$ where $g=(g_1,\ldots, g_n)$, and where the  $i^{\rm th}$ component of the vector $e_i$ is equal to 1 and all others components of $e_i$ are equal to zero.

We also used the following properties of the forces $(F^l,F^s)$ (see \cite[Section 2.5]{JoT}):
\begin{eqnarray}
|F^l(x,v)|&\le&2n\beta_1^l(1+|x|)^{-\alpha-1},\label{03a}\\
|F^s(x,v)|&\le&2n\beta_2^s(1+|x|)^{-\alpha-2},\label{03b}
\end{eqnarray}
\begin{eqnarray}
|F^l(x,v)-F^l(x',v')|&\le& {n\beta_1^l\over c}|v-v'|\sup_{\ep\in (0,1)}(1+|(1-\ep)x+\ep x'|)^{-\alpha-1}\label{03c}\\
&&+2n^{3\over 2}\beta_2^l|x-x'|\sup_{\ep\in (0,1)}(1+|(1-\ep)x+\ep x'|)^{-\alpha-2},\nonumber
\end{eqnarray}
\begin{eqnarray}
|F^s(x,v)-F^s(x',v')|&\le& {n\beta_2^s\over c}|v-v'|\sup_{\ep\in (0,1)}(1+|(1-\ep)x+\ep x'|)^{-\alpha-2}\label{03d}\\
&&+2n^{3\over 2}\beta_3^s|x-x'|\sup_{\ep\in (0,1)}(1+|(1-\ep)x+\ep x'|)^{-\alpha-3},\nonumber
\end{eqnarray}
for $(x,x',v,v')\in (\R^n)^4$, $\max(|v|,|v'|)\le c$.

\begin{proof}[Proof of Lemma \ref{lem_scatinit}]
We prove the existence and uniqueness of the solution $z_+$ (similarly one can prove the existence and uniqueness of $z_-$).
Set $C={2^{9\over 2}n^{3\over 2}\beta_1^l\sqrt{1-{|v|^2\over c^2}}\over \alpha |v|(1+{|x|\over \sqrt{2}}-|q|)^\alpha}$.
Let $\V_{w,C}$ be the complete metric space defined by 
$$
\V_{w,C}:=\{f\in C^1(\R,\R^n)\ |\ f(0)=0\textrm{ and }\sup_\R|\dot f|\le C,\ \sup_{\R}|w+\dot f|\le c\},
$$
endowed with the following norm $\|f\|_{\V}:=\sup_{\R}|\dot f|$. Note that $\V_{w,C}$ is a convex subset of $C^1(\R,\R^n)$.
We consider the integral equations
\begin{equation}
G_w(f)(t):=\int_0^t \dot G_w(f)(s) ds,\label{5.4a}
\end{equation}
\begin{equation}
\dot G_w(f)(t):=g\Big(g^{-1}(w)-\int_t^{+\infty}F^l(x+q+. w+f)(\tau)d\tau\Big)-w,\label{5.4b}
\end{equation}
for $f\in \V_{w,C}$ and for $t\in\R$. Then  we have (see also \eqref{01})
\begin{equation}
G_w(f)(0)=0,\ |w+\dot G_w(f)(t)|< c\textrm{ for }t\in  \R.\label{5.4c}
\end{equation}
We use the following estimate \eqref{5.5a}
\begin{eqnarray}
&&|x+q+\tau w+f(\tau)|\ge |x+\tau v|-|f(\tau)|-\tau|w-v|-|q|\nonumber\\
&\ge&{|x|\over \sqrt{2}}-|q|+\big({|v|\over \sqrt{2}}-C-|w-v|\big)|\tau|\ge{|x|\over \sqrt{2}}-|q|+{|v|\over 2\sqrt{2}}|\tau|,\label{5.5a}
\end{eqnarray}
for $\tau\in \R$ and $f\in \V_{w,C}$ (we used the estimate $|f(t)|\le C|t|$ for $t\in \R$ and for $f\in \V_{w,C}$, and we used that $x\cdot v=0$, and we used \eqref{5.1} and \eqref{5.1a}).
Using \eqref{03a} and \eqref{5.5a} we obtain that
\begin{eqnarray}
\int_t^{+\infty}|F^l(.w+x+q+f)(\tau)|d\tau &\le& 2n\beta_1^l\int_t^{+\infty}\big(1+{|x|\over \sqrt{2}}-|q|+{|v|\over 2\sqrt{2}}|\tau|\big)^{-\alpha-1}d\tau \nonumber\\
&\le&{2^{7\over 2}n\beta_1^l\over \alpha |v|(1+{|x|\over \sqrt{2}}-|q|)^\alpha},\label{5.6}
\end{eqnarray}
for $t\in \R$ and for $f\in \V_{w,C}$. We used the integral value
$\int_{-\infty}^{+\infty}(a+b\tau)^{-\alpha-1}d\tau={2\over b\alpha a^\alpha}$ for $a>0$ and $b>0$.
Then we also use the identity $|g^{-1}(v)|=|g^{-1}(w)|$ (see \eqref{5.1}) and we use \eqref{5.1a}, and we obtain
\begin{eqnarray}
&&\big|g^{-1}(w)-\int_t^{+\infty}(\ep F^l(. w+x+q+f_1)(\tau)d\tau+\eta F^l(.w+x+q+f_2)(\tau))d\tau\big|\nonumber\\
&\ge& |g^{-1}(v)|-{2^{7\over 2}n\beta_1^l\over \alpha |v|(1+{|x|\over \sqrt{2}}-|q|)^\alpha}\ge {|v|\over 2\sqrt{1-{|v|^2\over c^2}}},\label{5.7}
\end{eqnarray}
for $(f_1,f_2)\in \V_{w,C}^2$, and for $(\ep,\eta,t)\in (0,1)^2\times \R$, $\ep+\eta\le 1$.
Therefore combining \eqref{5.4b}, \eqref{5.6}, \eqref{5.7} and \eqref{02b} we obtain 
\begin{eqnarray}
|\dot G_w(f)(t)|&\le &{2^{7\over 2}n^{3\over 2}\beta_1^l\over \alpha |v|(1+{|x|\over \sqrt{2}}-|q|)^\alpha}\sup_{\ep\in (0,1)}{1\over \sqrt{1+{\big|g^{-1}(w)-\ep\int_t^{+\infty}F^l(. w+x+q+f)(\tau)d\tau\big|^2\over c^2}}}\nonumber\\
&\le &{2^{9\over 2}n^{3\over 2}\beta_1^l\sqrt{1-{|v|^2\over c^2}}\over \alpha |v|(1+{|x|\over \sqrt{2}}-|q|)^\alpha}.\label{5.8}
\end{eqnarray}
for $t\in \R$ and $f\in \V_{w,C}$.

Now let $(f_1,f_2)\in \V_{w,C}^2$. Then  using \eqref{02b}, \eqref{5.4b} and \eqref{5.7} we have
\begin{eqnarray}
|\dot G_w(f_1)(t)-\dot G_w(f_2)(t)|&\le&2n^{1\over 2}\sqrt{1-{|v|^2\over c^2}}\int_t^{+\infty}\big|F^l(. w+x+q+f_1)
\nonumber\\
&&-F^l(. w+x+q+f_2)\big|(\tau)d\tau\label{5.10}
\end{eqnarray}
for $t\in \R$. Using \eqref{03c} and  \eqref{5.5a} we have
\begin{eqnarray}
&&|F^l(.w+x+q+f_1)-F^l(. w+x+q+f_2)|(\tau)\nonumber\\
&\le& {n\beta_1^l\over c}(1+{|x|\over \sqrt{2}}-|q|+{|v|\over 2^{3\over 2}}|\tau|)^{-\alpha-1}\sup_{(0,+\infty)}|\dot f_1-\dot f_2|
\nonumber\\
&&+2n^{3\over 2}\beta_2^l(1+{|x|\over \sqrt{2}}-|q|+{|v|\over 2^{3\over 2}}|\tau|)^{-\alpha-2}|\tau|\sup_{s\in\R\b\{0\}}{|(f_1-f_2)(s)|\over |s|}\nonumber\\
&\le& \big({n\beta_1^l\over c}+{2^{5\over 2}n^{3\over 2}\over|v|}\beta_2^l\big)(1+{|x|\over \sqrt{2}}-|q|+{|v|\over 2^{3\over 2}}|\tau|)^{-\alpha-1}\|f_1-f_2\|_{\V}
\label{5.11}
\end{eqnarray}
for $\tau\in\R$.
Therefore we obtain
\begin{equation}
|\dot G_w(f_1)(t)-\dot G_w(f_2)(t)|
\le
{n^{3\over 2}2^{7\over 2}\sqrt{1-{|v|^2\over c^2}}\big({\beta_1^l\over c}+{2^{5\over 2}n^{1\over 2}\over|v|}\beta_2^l\big)\over \alpha|v|(1+{|x|\over \sqrt{2}}-|q|)^{\alpha}}\|f_1-f_2\|_{\V},\label{5.12a}
\end{equation}
for $t\in \R$.

From \eqref{5.4c}, \eqref{5.8}, \eqref{5.12a} and \eqref{5.1a} it follows that the operator $G_{w}$ is a ${1\over 2}$-contraction map from $\V_{w,C}$ to $\V_{w,C}$. Set $z_+(w,x+q,t)=x+q+tw+f_{w,x+q}(t)$ for $t\in \R$, where $f_{w,x+q}$ denotes the unique fixed point of $G_w$ in $\V_{w,C}$. Then $z_+(w,x+q,.)$ satisfies \eqref{5.3}, \eqref{5.3a}, \eqref{5.3b}.
\end{proof}

Before proving Lemma \ref{lem_decomp} we recall the following standard result. For sake of consistency we provide a proof of Lemma \ref{lem:comp} at the end of this Section.

\begin{lemma}
\label{lem:comp}
Let $x\in C^2(\R,\R^n)$ satisfy equation \eqref{1.1} and let $z\in C^2(\R,\R^n)$ satisfy equation \eqref{5.3}. Assume that there exists $v\in \B(0,c)$, $v\not=0$, so that $\dot x(t)\to v$ and $\dot z(t)\to v$ as $t\to +\infty$. Then 
\begin{equation}
\sup_{(0,+\infty)}|x-z|<\infty\textrm{ and }\sup_{t\in (0,+\infty)}(1+t)^{-1}|\dot x-\dot z|(t)<\infty.\label{3.11}
\end{equation}
\end{lemma}

\begin{proof}[Proof of Lemma \ref{lem_decomp}]
We need the following preliminary estimate \eqref{2.5b}.
From the formula $g(\tau)=g(0)+\int_0^\tau \dot g(s)ds$ for $g\in C^1(\R,\R^n)$ it follows that ,   
\begin{equation}
|(f_1-f_2)(\tau)|\le \sup_{(-\infty,0)}|f_1-f_2|+|\tau|\sup_{\R}|\dot f_1-\dot f_2|\textrm{ for }\tau\in \R, \label{2.5a}
\end{equation}
and for $(f_1,f_2)\in M_{r,v_-}^2$.
Hence
\begin{eqnarray}
|z_-(v_-,\tau)+x+f(\tau)|&\ge& |x+\tau v_-|-|z_-(v_-,\tau)-\tau v_-|-|f(\tau)|\nonumber\\
&\ge&{|x|\over \sqrt{2}}-r+|\tau|\Big({|v_-|\over \sqrt{2}}-{2^{9\over 2}n^{3\over 2}\beta_1^l\sqrt{1-{|v_-|^2\over c^2}}\over \alpha |v_-|}-r\Big)\nonumber\\
&\ge & {|x|\over \sqrt{2}}-r+|\tau|\Big({|v_-|\over 2\sqrt{2}}-r\Big),\label{2.5b}
\end{eqnarray}
for $(f,\tau)\in M_{r,v_-}\times \R$ and for any $x\in \R^n$ so that $x\cdot v_-=0$.  
We used \eqref{5.3b} (for "$(x,w,v,q)=(0,v_-,v_-,0)$"), the inequality $|x+\tau v_-|\ge {|x|\over \sqrt{2}}+|\tau|{|v_-|\over \sqrt{2}}$ ($x\cdot v_-=0$) and \eqref{2.5a} (for $(f_1,f_2)=(f,0)$ and $\|f\|\le r$) and the condition $|v_-|\ge \mu^l$.

Hence the integral $\int_{-\infty}^{+\infty}F(z_-(v_-,.)+x_-+f)(\tau)d\tau$ is absolutely convergent for any $f\in M_{r,v_-}$. And when $y_-\in M_{r,v_-}$ is a fixed point for $A$ then
$z_-(v_-,.)+x_-+y_-$ satisfies equation \eqref{1.1} (see \eqref{2.1a} and \eqref{2.6}) and 
$\dot z_-(v_-,t)+\dot y_-(t)=g\big(g^{-1}(v_-)+\int_{-\infty}^tF(z_-(v_-,.)+x_-+y_-)(\tau)d\tau\big)\to a(v_-,x_-)\textrm{ as }t\to +\infty,$
where $a(v_-,x_-)$ is defined in \eqref{300a}. Then from Lemma \ref{lem:comp} it follows that $\sup_{(0,+\infty)}|z_-(v_-,.)+x_-+y_--z_+(a(v_-,x_-),.)|<+\infty$ and $\sup_{t\in(0,+\infty)}(1+t)^{-1}|\dot z_-(v_-,t)+\dot y_-(t)-\dot z_+(a(v_-,x_-),t)|<\infty$. Using these latter estimates and $y_-\in M_{r,v_-}$, and using \eqref{03b}, \eqref{03c} and \eqref{2.5b} we obtain that the integral on the right hand side of \eqref{3.13b} is absolutely convergent. 
Then the decomposition \eqref{3.13} follows from the equality $A(y_-)=y_-$ and \eqref{2.1a} and \eqref{2.6} and straightforward computations. 
\end{proof}

\begin{proof}[Proof of Lemma \ref{lem:comp}]
Note that from \eqref{1.1}, \eqref{5.3} and the property $\lim_{t\to +\infty}\dot x(t)=\lim_{t\to +\infty}\dot z(t)=v\not=0$ it follows that
\begin{equation}
\dot x(t)=g\big(g^{-1}(v)-\int_t^{+\infty}F(x)(s)ds\big),\ \dot z(t)=g\big(g^{-1}(v)-\int_t^{+\infty}F^l(z)(s)ds\big),\label{3.11a}
\end{equation}
for $t\ge 0$, and that there exists $\ep>0$ so that
\begin{equation}
1+|\eta x(t)+(1-\eta)z(t)|\ge \ep(1+t), \textrm{ for }(t, \eta)\in [0,+\infty)\times [0,1].\label{3.11c}
\end{equation}
Using \eqref{3.11c} for $\eta=1$ and using \eqref{03b} we obtain 
\begin{equation}
\int_t^{+\infty}|F^s(x)(\tau)|d\tau\le 2n\beta_2\ep^{-\alpha-2}\int_t^{+\infty}(1+\tau)^{-\alpha-2}d\tau
={2n\beta_2\ep^{-\alpha-2}\over (\alpha+1)(1+t)^{\alpha+1}},\label{3.11d}
\end{equation}
for $t\ge 0$.
Thus using \eqref{3.11a}, \eqref{02b} and \eqref{3.11d} we have 
\begin{equation}
|\dot x(t)-\dot z(t)|\le
n^{1\over 2}\int_t^{+\infty}|F^l(x)-F^l(z)|(\tau)ds+{2n^{3\over 2}\beta_2\ep^{-\alpha-2}\over (\alpha+1)(1+t)^{\alpha+1}},\label{3.11e}
\end{equation}
for $t\ge 0$.
Using \eqref{3.11c} for $\eta=0,1,$ and using \eqref{03a} we obtain 
\begin{equation}
\int_t^{+\infty}(|F^l(x)|+|F^l(z)|)(\tau)d\tau\le {4\ep^{-\alpha-1}n\beta_1^l\over \alpha(1+t)^\alpha},\label{3.11ff}
\end{equation}
for $t\ge 0$,
and using \eqref{3.11e} we obtain
\begin{equation}
|\dot x(t)-\dot z(t)|
\le{4n^{3\over 2}\ep^{-\alpha-1}\beta_1^l+2n^{3\over 2}\beta_2\ep^{-\alpha-2}\over \alpha(1+t)^\alpha},\textrm{ for }t\ge 0.\label{3.11gg}
\end{equation}

We assume without loss of generality that $\alpha\not={1\over m}$ for $m\in \N$, $m\not=0$. If $\alpha={1\over m}$ for some $m\in \N$, $m\not=0$, then 
just replace $\alpha$ by $\alpha'\in ({1\over m+1}, {1\over m})$. We denote by $\lfloor x\rfloor$ the integer part of any real number $x$.
We prove by induction that for $n=1\ldots \lfloor {1\over \alpha}\rfloor$ there exists a positive constant $C_n$ so that
\begin{equation}
|\dot x(t)-\dot z(t)|
\le C_n(1+t)^{-n\alpha},\textrm{ for }t\ge 0.\label{3.11f}
\end{equation}
The estimate \eqref{3.11f} is proved for $n=1$ by \eqref{3.11gg}. 
Let $n=1\ldots \lfloor {1\over \alpha}\rfloor$. Integrating \eqref{3.11f} over $(0,t)$ we have 
\begin{equation}
|x(t)-z(t)|
\le C_n'+{C_n\over |1-n\alpha|}(1+t)^{1-n\alpha},\textrm{ for }t\ge 0,\label{3.11g}
\end{equation}
and for some constant $C_n'$.
Using \eqref{3.11e}, \eqref{3.11f}, \eqref{3.11g}, \eqref{03c} and \eqref{3.11c} we have
\begin{eqnarray}
|\dot x(t)-\dot z(t)|
&\le&\int_t^{+\infty}\sup_{\eta\in (0,1)}{{n^{3\over 2}\beta_1^l\over c}|\dot x(s)-\dot z(s)|\over (1+|\eta x(s)+(1-\eta)z(s)|)^{\alpha+1}}ds\nonumber\\
&&+\int_t^{+\infty}\sup_{\eta\in (0,1)}{2n^2\beta_2|x(s)-z(s)|\over (1+|\eta x(s)+(1-\eta)z(s)|)^{\alpha+2}}ds\\
&&+{2n^{3\over 2}\beta_2\ep^{-\alpha-2}\over (\alpha+1)(1+t)^{\alpha+1}}\nonumber
\end{eqnarray}
\begin{eqnarray}
&\le&n^{3\over 2}\ep^{-\alpha-2}\big({\beta_1^lC_n\ep\over c}+2{n^{1\over 2}C_n\beta_2\over  |1-n\alpha|}\big)\int_t^{+\infty}(1+s)^{-1-(n+1)\alpha}ds\nonumber\\
&&+2n^2\beta_2C_n'\ep^{-\alpha-2}\int_t^{+\infty}(1+s)^{-\alpha-2}ds
+{2n^{3\over 2}\beta_2\ep^{-\alpha-2}\over (\alpha+1)(1+t)^{\alpha+1}}\nonumber\\
&\le&{n^{3\over 2}\ep^{-\alpha-2}\big({\beta_1^lC_n\ep\over c}+2{n^{1\over 2}C_n\beta_2\over  |1-n\alpha|}\big)\over (n+1)\alpha(1+t)^{(n+1)\alpha}}
+{2n^{3\over 2}\ep^{-\alpha-2}\beta_2(1+n^{1\over 2}C_n')\over (\alpha+1)(1+t)^{\alpha+1}},\label{3.11h}
\end{eqnarray}
for $t\ge 0$. This proves \eqref{3.11f} for "$n+1$''. The induction step is proved.
Then \eqref{3.11f}, \eqref{3.11g} also hold for $n=\lfloor {1\over \alpha}\rfloor+1$, and using $(\lfloor {1\over \alpha}\rfloor+1)\alpha>1$, we obtain \eqref{3.11}.
\end{proof}

\section{Proof of Lemmas \ref{lem_cont} and \ref{lem_cont2}}
\label{sec_cont}
\begin{proof}[Proof of Lemma \ref{lem_cont}]
We shorten $z_-(v_-,.)$ to $z_-$ in this paragraph.
We first prove the estimates \eqref{2.14}, \eqref{2.15} and \eqref{2.18} given below.
Let $f\in M_{r,v_-}$.
Using \eqref{03b} and \eqref{2.5b} we have 
\begin{equation}
 |F^s(z_-+\ep x_-+f)(\tau)|
\le {2n\beta_2\over\big(1+{\ep|x_-|\over \sqrt{2}}-r+|\tau|\big({|v_-|\over 2\sqrt{2}}-r\big)\big)^{\alpha+2}},\label{2.8}
\end{equation}
for $\tau\in \R$ and for $\ep\in [0,1]$.
Integrating both sides of \eqref{2.8} over $(-\infty,t)$ we obtain
\begin{equation}
\int_{-\infty}^t \big|F^s(z_-+\ep x_-+f)(\tau)\big|d\tau
\le {4n\beta_2\over (\alpha+1)\big({|v_-|\over 2\sqrt{2}}-r\big)\big(1+{\ep|x_-|\over \sqrt{2}}-r\big)^{\alpha+1}},\label{2.8b}
\end{equation} 
for $t\in \R$ and for $\ep\in [0,1]$. 
Similarly using \eqref{03a} instead of \eqref{03b} we have 
\begin{equation}
\int_{-\infty}^t \big|F^l(z_-+\ep x_-+f)(\tau)\big|d\tau
\le {4n\beta_1^l\over \alpha\big({|v_-|\over 2\sqrt{2}}-r\big)\big(1+{\ep|x_-|\over \sqrt{2}}-r\big)^\alpha},\label{2.9b}
\end{equation} 
for $t\in \R$ and for $\ep\in [0,1]$. 
We combine \eqref{2.8b} and \eqref{2.9b}, and we use \eqref{lc0}, and we have 
\begin{eqnarray}
&&|g^{-1}(v_-)+\eta_1\int_{-\infty}^t\big(F(z_-+x_-+f_1)(\tau)
+(\eta_2F^l+\eta_2\mu F^s)(z_-+\ep x_-+f_2)(\tau)d\tau|\nonumber\\
&&\ge |g^{-1}(v_-)|-{8n\max(\beta_1^l,\beta_2)\over \alpha\big({|v_-|\over 2\sqrt{2}}-r\big)\big(1-r\big)^{\alpha+1}}\ge {1\over 2}|g^{-1}(v_-)|.\label{2.9c}
\end{eqnarray}
for $(f_1,f_2)\in M_{r,v_-}^2$ and for $(t,\eta_1,\eta_2,\mu,\ep)\in \R\times[0,1]^4$ so that $\eta_1+\eta_2\le 1$.
Then from \eqref{2.6}, \eqref{2.9c} and \eqref{02b} it follows that
\begin{equation}
|\dot A(f)(t)|\le 
2n^{1\over 2}\sqrt{1-{|v_-|^2\over c^2}}\Big|\int_{-\infty}^t\big(F(z_-+x_-+f)(\tau)-F^l(z_-)(\tau)\big)d\tau\Big|,\label{2.7}
\end{equation}
for $t\in \R$.
Using \eqref{2.8} we have
\begin{equation}
\int_{-\infty}^t \big|F^s(z_-+x_-+f)(\tau)\big|d\tau
\le {2n\beta_2\over (\alpha +1)\big({|v_-|\over 2\sqrt{2}}-r\big)\Big(1+{|x_-|\over \sqrt{2}}-r-t\big({|v_-|\over 2^{3\over 2}}-r\big)\Big)^{\alpha+1}},\label{2.8a}
\end{equation} 
for $t\le 0$.
Using \eqref{03c} and \eqref{2.5b} (with $''x''=(1-\ep +\ep\mu)x_-$, $(\ep,\mu) \in [0,1]^2$) we obtain
\begin{eqnarray}
\big|F^l(z_-+x_-+f_1)(\tau)-F^l(z_-+\mu x_-+f_2)(\tau)\big|\nonumber\\
\le{{n\beta_1^l\over c}|\dot f_1-\dot f_2|(\tau)\over\big(1-r+\big({|v_-|\over 2^{3\over 2}}-r\big)|\tau|)^{\alpha+1}}+{2n^{3\over 2}\beta_2\big((1-\mu)|x_-|+|f_1-f_2|(\tau)\big)\over \big(1-r+\big({|v_-|\over 2^{3\over 2}}-r\big)|\tau|)^{\alpha+2}}\label{2.9}
\end{eqnarray}
for $(f_1,f_2)\in M_{r,v_-}^2$ and for $(\tau,\mu) \in \R\times[0,1]$.
We integrate \eqref{2.9} over $(-\infty,t)$, and we use the estimates $|f(\tau)|\le r$ and 
$|\dot f(\tau)|\le r(1-r+({|v_-|\over 2\sqrt{2}}-r)|\tau|)^{-1}$ for $\tau\le 0$, and we obtain
\begin{eqnarray}
&&\int_{-\infty}^t\big|F^l(z_-+x_-+f)(\tau)-F^l(z_-)(\tau)\big|d\tau\nonumber\\
&\le &({n\beta_1^lr\over c}+2n^{3\over 2}\beta_2(|x_-|+r))\int_{-\infty}^t\big(1-r+\big({|v_-|\over 2\sqrt{2}}-r\big)|\tau|\big)^{-\alpha-2}d\tau\nonumber\\
&\le&{{nr\beta_1^l\over c}+2n^{3\over 2}\beta_2(|x_-|+r)\over (\alpha+1)\big({|v_-|\over 2\sqrt{2}}-r\big)\big(1-r-\big({|v_-|\over 2\sqrt{2}}-r\big)t\big)^{\alpha+1}},\label{2.13}
\end{eqnarray}
for $t\le 0$.

Combining \eqref{2.7}, \eqref{2.13}, \eqref{2.8a} we obtain
\begin{eqnarray}
|\dot A(f)(t)|
&\le& 
{2n^{3\over 2}\left({1-{|v_-|^2\over c^2}}\right)^{1\over 2}\big({r\beta_1^l\over c}+2\beta_2(n^{1\over 2}(|x_-|+r)+1)\big)\over (\alpha+1)\big({|v_-|\over 2\sqrt{2}}-r\big)\big(1-r-\big({|v_-|\over 2\sqrt{2}}-r\big)t\big)^{\alpha+1}},\label{2.14}\\
|A(f)(t)|&\le& {2n^{3\over 2}\left({1-{|v_-|^2\over c^2}}\right)^{1\over 2}\big({r\beta_1^l\over c}+2\beta_2(n^{1\over 2}(|x_-|+r)+1)\big)\over \alpha(\alpha+1)\big({|v_-|\over 2\sqrt{2}}-r\big)^2\big(1-r-\big({|v_-|\over 2\sqrt{2}}-r\big)t\big)^\alpha},\label{2.15}
\end{eqnarray}
for $t\le 0$.

Let $t\ge 0$ and $|v_-|> 2\sqrt{2}r$, $r<1$.
Integrating \eqref{2.9} over $(0,t)$ and using the estimate  $\sup_{(0,+\infty)}|\dot f|\le r$ and \eqref{2.5a} (for $(f_1,f_2)=(f,0)$ and $\|f\|\le r$), we obtain
\begin{eqnarray}
&&\int_0^t\big|F^l(z_-+x_-+f)(\tau)-F^l(z_-)(\tau)\big|d\tau\nonumber\\
&\le&{nr\beta_1^l\over c}\int_0^t\big(1-r+({|v_-|\over 2\sqrt{2}}-r)|\tau|\big)^{-\alpha-1}d\tau\nonumber\\
&&+2n^{3\over 2}\beta_2\int_0^t\big(1-r+({|v_-|\over 2\sqrt{2}}-r)|\tau|\big)^{-\alpha-2}(|x_-|+r+r|\tau|)d\tau\nonumber\\
&\le&{n\over ({|v_-|\over 2\sqrt{2}}-r) (1-r)^\alpha}\big({{r\beta_1^l\over c}\over \alpha}
+{2n^{1\over 2}\beta_2(|x_-|+r)\over (\alpha+1)(1-r)}
+{2n^{1\over 2}\beta_2r\over \alpha({|v_-|\over 2\sqrt{2}}-r)}\big).
\label{2.17}
\end{eqnarray}
Hence combining \eqref{2.7}, \eqref{2.8b}, \eqref{2.17} and \eqref{2.13} (for "$t=0$") we obtain
\begin{equation}
|\dot A(f)(t)|
\le{2n^{3\over 2}\left({1-{|v_-|^2\over c^2}}\right)^{1\over 2}\over({|v_-|\over 2\sqrt{2}}-r) (1-r)^\alpha}\Big({{r\beta_1^l\over c}+4\beta_2(n^{1\over 2}(|x_-|+r)+1)\over (\alpha+1)\big(1-r\big)}+{r\beta_1^l\over c\alpha}
+{2n^{1\over 2}\beta_2r\over \alpha({|v_-|\over 2\sqrt{2}}-r) }\Big).\label{2.18}
\end{equation}
Then estimate \eqref{lc1} follows from \eqref{2.14}, \eqref{2.15} and \eqref{2.18}.

Now we prove the estimates \eqref{2.26}, \eqref{2.27} and \eqref{2.29} given below. Estimate \eqref{lc2} follows from those latter estimates.
Let $|v_-|> 2\sqrt{2}r$, $r< 1$, and let $(f_1,h_1)$ and $(f_2,h_2)\in M_{r,v_-}$. From \eqref{03d} and \eqref{2.5b} it follows that
\begin{eqnarray}
\big|F^s(z_-+x_-+f_1)(\tau)-F^s(z_-+x_-+f_2)(\tau)\big|\nonumber\\
\le{{n\beta_2\over c}|\dot f_1-\dot f_2|(\tau)\over\big(1-r+{|x_-|\over \sqrt{2}}+\big({|v_-|\over 2^{3\over 2}}-r\big)|\tau|)^{\alpha+2}}
+{2n^{3\over 2}\beta_3^s|f_1-f_2|(\tau)\over \big(1-r+{|x_-|\over \sqrt{2}}+\big({|v_-|\over 2^{3\over 2}}-r\big)|\tau|)^{\alpha+3}}\label{2.9ab}
\end{eqnarray}
for $\tau \in \R$.
 Note that
\begin{eqnarray}
&&\dot A(f_1)(t)-\dot A(f_2)(t)
=g\Big(g^{-1}(v_-)+\int_{-\infty}^tF(z_-+x_-+f_1)(\tau)d\tau\Big)\nonumber\\
&&-g\Big(g^{-1}(v_-)+\int_{-\infty}^tF(z_-+x_-+f_2)(\tau)d\tau\Big).\label{2.21}
\end{eqnarray}
Hence using \eqref{02b} we obtain
\begin{equation}
|\dot A(f_1)(t)-\dot A(f_2)(t)|\le 2n^{1\over 2}\sqrt{1-{|v_-|^2\over c^2}}J(t),\label{2.22}
\end{equation}
\begin{equation}
J(t):=\int_{-\infty}^t|F(z_-+x_-+f_2)(\tau)
-F(z_-+x_-+f_1)(\tau)|d\tau,\label{2.23}
\end{equation} 
for $t\in \R$.
We integrate \eqref{2.9} and \eqref{2.9ab} over $(-\infty,t)$, and we use the estimates $|(f_1- f_2)(\tau)|\le \|f_1-f_2\|$ and 
$|(\dot f_1-\dot f_2)(\tau)|\le (1-r+({|v_-|\over 2\sqrt{2}}-r)|\tau|)^{-1} \|f_1-f_2\|$ for $\tau\le 0$, and  we have
\begin{equation}
J(t)\le\|f_1-f_2\|\int_{-\infty}^t\Big({\big({n\beta_1^l\over c}
+2n^{3\over 2}\beta_2\big)\over \big(1-r+({|v_-|\over 2\sqrt{2}}-r)|\tau|\big)^{\alpha+2}}
+{\big({n\beta_2\over c}
+2n^{3\over 2}\beta_3^s\big)\over \big(1-r+({|v_-|\over 2\sqrt{2}}-r)|s|\big)^{\alpha+3}}\Big)d\tau.\label{2.25}
\end{equation}
Let $t\le 0$. We also use \eqref{2.22} and we obtain
\begin{eqnarray}
&&|\dot A(f_1)(t)-\dot A(f_2)(t)|\le
{2n^{3\over 2}\Big({1-{|v_-|^2\over c^2}}\Big)^{1\over 2}\|f_1-f_2\|\over({|v_-|\over 2\sqrt{2}}-r)\big(1-r+({|v_-|\over 2\sqrt{2}}-r)|t|\big)^{\alpha+1}}\nonumber\\
&&\times\Big({{\beta_1^l\over c}
+2n^{1\over 2}\beta_2\over (\alpha+1)}+{{\beta_2\over c}
+2n^{1\over 2}\beta_3^s\over (\alpha+2)\big(1-r+({|v_-|\over 2\sqrt{2}}-r)|t|\big)}\Big),\label{2.26}
\end{eqnarray}
\begin{eqnarray}
&&|A(f_1)(t)-A(f_2)(t)|\le
{2n^{3\over 2}\Big({1-{|v_-|^2\over c^2}}\Big)^{1\over 2}\|f_1-f_2\|\over(\alpha+1)({|v_-|\over 2\sqrt{2}}-r)^2\big(1-r+({|v_-|\over 2\sqrt{2}}-r)|t|\big)^\alpha}\nonumber\\
&&\times\Big({{\beta_1^l\over c}
+2n^{1\over 2}\beta_2\over \alpha}+{{\beta_2\over c}
+2n^{1\over 2}\beta_3^s\over (\alpha+2)\big(1-r+({|v_-|\over 2\sqrt{2}}-r)|t|\big)}\Big).\label{2.27}
\end{eqnarray}

Let $t\ge 0$ and $|v_-|>2\sqrt{2}r$, $r<1$, and let $(f_1,f_2)\in M_{r, v_-}^2$.
We integrate \eqref{2.9} and \eqref{2.9ab} over $(0,t)$, and we use the estimates \eqref{2.5a} and $|\dot f_1(\tau)-\dot f_2(\tau)|\le \|f_1-f_2\|$ for $\tau \in \R$, and we have 
\begin{eqnarray}
J(t)-J(0)&=&\int_0^t|F(x_-+z_-+f_2)(\tau)-F(x_-+z_-+f_1)(\tau)|d\tau\nonumber\\
&\le&\int_0^t{{n\beta_1^l\over c}\|f_1-f_2\|d\tau\over \big(1-r+({|v_-|\over 2\sqrt{2}}-r)\tau\big)^{\alpha+1}}
+
\int_0^t{2n^{3\over 2}\beta_2(1+\tau)\|f_1-f_2\|d\tau\over \big(1-r+({|v_-|\over 2\sqrt{2}}-r)\tau\big)^{\alpha+2}}\nonumber\\
&&+\int_0^t{{n\beta_2\over c}\|f_1-f_2\|d\tau\over \big(1-r+({|v_-|\over 2\sqrt{2}}-r)\tau\big)^{\alpha+2}}
+
\int_0^t{2n^{3\over 2}\beta_3^s(1+\tau)\|f_1-f_2\|d\tau\over \big(1-r+({|v_-|\over 2\sqrt{2}}-r)\tau\big)^{\alpha+3}}\nonumber\\
&\le&{n\|f_2-f_1\|\over (1-r)^\alpha ({|v_-|\over 2\sqrt{2}}-r)}\Big({\beta_1^l\over c\alpha}+{\beta_2\over c(\alpha+1)(1-r)}
+{2n^{1\over 2}\over(1-r)}
\big({\beta_2\over \alpha+1}\nonumber\\
&&+{\beta_3^s\over (\alpha+2)(1-r)}\big)+{2n^{1\over 2}\over ({|v_-|\over 2\sqrt{2}}-r)}
\big({\beta_2\over \alpha}+{\beta_3^s\over (\alpha+1)(1-r)}\big)\Big).\label{2.28}
\end{eqnarray}
Using \eqref{2.22}, \eqref{2.25} (with ``$t=0$'') and \eqref{2.28} we obtain
\begin{eqnarray}
&&|\dot A(f_1)(t)-\dot A(f_2)(t)|\le
{2n^{3\over 2}\Big({1-{|v_-|^2\over c^2}}\Big)^{1\over 2}\|f_1-f_2\|\over ({|v_-|\over 2\sqrt{2}}-r)\big(1-r\big)^\alpha}\nonumber\\
&&\times\Big[{\beta_1^l
\over c(\alpha+1)\big(1-r\big)}+{\beta_2
\over c(\alpha+2)\big(1-r\big)^2}\nonumber\\
&&+\big({\beta_1^l\over c\alpha}+{\beta_2\over c(\alpha+1)(1-r)}\big)+{4n^{1\over 2}\over(1-r)}
\big({\beta_2\over \alpha+1}+{\beta_3^s\over (\alpha+2)(1-r)}\big)\nonumber\\
&&+{2n^{1\over 2}\over({|v_-|\over 2\sqrt{2}}-r)}
\big({\beta_2\over \alpha}+{\beta_3^s\over (\alpha+1)(1-r)}\big)
\Big].
\label{2.29}
\end{eqnarray}
\end{proof}

\begin{proof}[Proof of Lemma \ref{lem_cont2}]
In this paragraph we shorten $z_-(v_-,x_-,.)$ to $z_-$.

Similarly to \eqref{2.5b} we have
\begin{eqnarray}
|z_-(\tau)+f(\tau)|&\ge& {|x_-|\over \sqrt{2}}-r+|\tau|\Big({|v_-|\over \sqrt{2}}-{2^{9\over 2}n^{3\over 2}\beta_1^l\sqrt{1-{|v_-|^2\over c^2}}\over \alpha |v_-|(1+{|x_-|\over \sqrt{2}})^\alpha}-r\Big)\nonumber\\
&\ge & {|x_-|\over \sqrt{2}}-r+|\tau|\Big({|v_-|\over 2\sqrt{2}}-r\Big).\label{4.5b}
\end{eqnarray}
for $\tau\in \R$ and for $f\in M_{r,v_-,x_-}$. We used \eqref{5.3b} (for "$(x,w,v,q)=(x_-,v_-,v_-,0)$") and the condition \eqref{5.1a}.
Using \eqref{03b}, \eqref{03a} and \eqref{4.5b} we have 
\begin{eqnarray}
|F^s(z_-+f)(\tau)|
&\le& 2n\beta_2\Big(1+{|x_-|\over \sqrt{2}}-r+|\tau|\big({|v_-|\over 2\sqrt{2}}-r\big)\Big)^{-\alpha-2},\label{4.9}\\
|F^l(z_-+f)(\tau)|
&\le& 2n\beta_1^l\Big(1+{|x_-|\over \sqrt{2}}-r+|\tau|\big({|v_-|\over 2\sqrt{2}}-r\big)\Big)^{-\alpha-1},\label{4.9b}
\end{eqnarray}
for $\tau\in \R$.
Hence using also condition \eqref{lc2_0} we obtain 
\begin{eqnarray}
&&|g^{-1}(v_-)+\int_{-\infty}^t\big(\eta F(z_-+f_1)+(1-\eta)(F^l+\mu F^s)(z_-+f_2)\big)(\tau)d\tau|\nonumber\\
&\ge&|g^{-1}(v_-)|-{4n\over \big({|v_-|\over 2^{3\over 2}}-r\big)(1+{|x_-|\over \sqrt{2}}-r)^\alpha}\big({\beta_1^l\over \alpha}+{\beta_2\over (\alpha+1)(1+{|x_-|\over \sqrt{2}}-r)}\big)\nonumber\\
&\ge&{|g^{-1}(v_-)|\over 2},\label{4.9e}
\end{eqnarray}
for $(t,\eta,\mu)\in \R\times [0,1]^2$ and for $(f_1,f_2)\in M_{r,v_-,x_-}^2$.

Using \eqref{4.5b}, \eqref{03c} and \eqref{03d} we have
\begin{eqnarray}
&&|F^l(z_-+f_1)-F^l(z_-+f_2)|(\tau)
\le{{n\beta_1^l\over c}|(\dot f_1-\dot f_2)(\tau)|\over \big(1-r+{|x_-|\over \sqrt{2}}+({|v_-|\over 2^{3\over 2}}-r)|\tau|\big)^{\alpha+1}}\nonumber\\
&&+{2n^{3\over 2}\beta_2|(f_1-f_2)(\tau)|\over \big(1-r+{|x_-|\over \sqrt{2}}+({|v_-|\over 2^{3\over 2}}-r)|\tau|\big)^{\alpha+2}}
,\label{4.9c}
\end{eqnarray}
\begin{eqnarray}
&&\big|F^s(z_-+f_1)(\tau)-F^s(z_-+f_2)(\tau)\big|
\le{{n\beta_2\over c}|\dot f_1-\dot f_2|(\tau)\over\big(1-r+{|x_-|\over \sqrt{2}}+\big({|v_-|\over 2^{3\over 2}}-r\big)|\tau|)^{\alpha+2}}\nonumber\\
&&+{2n^{3\over 2}\beta_3^s|f_1-f_2|(\tau)\over \big(1-r+{|x_-|\over \sqrt{2}}+\big({|v_-|\over 2^{3\over 2}}-r\big)|\tau|)^{\alpha+3}},
\label{4.9ab}
\end{eqnarray}
for $\tau\in \R$. 

Then the proof of the following estimates \eqref{4.14}, \eqref{4.15}, \eqref{4.18},  \eqref{4.26}, \eqref{4.27} and \eqref{4.29}  is similar to the proof of the estimates \eqref{2.14}, \eqref{2.15}, \eqref{2.18},  \eqref{2.26}, \eqref{2.27} and \eqref{2.29}, and we have  
\begin{eqnarray}
|\dot\A(f)(t)|
&\le&{ 2n^{3\over 2}\left({1-{|v_-|^2\over c^2}}\right)^{1\over 2}
\big({r\beta_1^l\over c}+2\beta_2(n^{1\over 2}r+1)\big)\over (\alpha+1)\big({|v_-|\over 2\sqrt{2}}-r\big)\big(1-r+{|x_-|\over \sqrt{2}}-\big({|v_-|\over 2\sqrt{2}}-r\big)t\big)^{\alpha+1}},\label{4.14}\\
|\A(f)(t)|
&\le& {2n^{3\over 2}\left({1-{|v_-|^2\over c^2}}\right)^{1\over 2}
\big({r\beta_1^l\over c}+2\beta_2(n^{1\over 2}r+1)\big)\over \alpha(\alpha+1)\big({|v_-|\over 2\sqrt{2}}-r\big)^2\big(1-r+{|x_-|\over \sqrt{2}}-\big({|v_-|\over 2\sqrt{2}}-r\big)t\big)^\alpha},\label{4.15}
\end{eqnarray}
\begin{eqnarray}
&&|\dot\A(f_1)(t)-\dot\A(f_2)(t)|\le
{2n^{3\over 2}\Big({1-{|v_-|^2\over c^2}}\Big)^{1\over 2}\|f_1-f_2\|_*\over({|v_-|\over 2\sqrt{2}}-r)\big(1-r+{|x_-|\over \sqrt{2}}+({|v_-|\over 2\sqrt{2}}-r)|t|\big)^{\alpha+1}} \nonumber\\
&&\times\Big({{\beta_1^l\over c}
+2n^{1\over 2}\beta_2\over (\alpha+1)}+{{\beta_2\over c}
+2n^{1\over 2}\beta_3^s\over (\alpha+2)\big(1-r+{|x_-|\over \sqrt{2}}+({|v_-|\over 2\sqrt{2}}-r)|t|\big)}\Big),\label{4.26}\\
&&|\A(f_1)(t)-\A(f_2)(t)|\le
{2n^{3\over 2}\Big({1-{|v_-|^2\over c^2}}\Big)^{1\over 2}\|f_1-f_2\|_*\over(\alpha+1)({|v_-|\over 2\sqrt{2}}-r)^2\big(1-r+{|x_-|\over \sqrt{2}}+({|v_-|\over 2\sqrt{2}}-r)|t|\big)^\alpha} \nonumber\\
&&\times\Big({{\beta_1^l\over c}
+2n^{1\over 2}\beta_2\over \alpha}+{{\beta_2\over c}
+2n^{1\over 2}\beta_3^s\over (\alpha+2)\big(1-r+{|x_-|\over \sqrt{2}}+({|v_-|\over 2\sqrt{2}}-r)|t|\big)}\Big),\label{4.27}
\end{eqnarray}
for $t\le 0$, and
\begin{equation}
|\dot\A(f)(t)|
\le {2n^{3\over 2}\left({1-{|v_-|^2\over c^2}}\right)^{1\over 2}\over ({|v_-|\over 2\sqrt{2}}-r) (1-r+{|x_-|\over \sqrt{2}})^\alpha}
\Big({{r\beta_1^l\over c}+4\beta_2(n^{1\over 2}r+1)\over (\alpha+1)
\big(1-r+{|x_-|\over \sqrt{2}}\big)}+{r\beta_1^l\over c \alpha}
+{2n^{1\over 2}\beta_2r\over \alpha({|v_-|\over 2\sqrt{2}}-r)}\Big),\label{4.18}
\end{equation}
\begin{eqnarray}
&&|\dot\A(f_1)(t)-\dot\A(f_2)(t)|\le
{2n^{3\over 2}\Big({1-{|v_-|^2\over c^2}}\Big)^{1\over 2}\|f_1-f_2\|_*\over ({|v_-|\over 2\sqrt{2}}-r)(1-r+{|x_-|\over \sqrt{2}})^\alpha}\nonumber\\
&&\times\Big[{\beta_1^l
\over c(\alpha+1)\big(1-r+{|x_-|\over \sqrt{2}}\big)}+{\beta_2\over c(\alpha+2)\big(1-r+{|x_-|\over \sqrt{2}}\big)^2}\nonumber\\
&&+\big({\beta_1^l\over c\alpha}+{\beta_2\over c(\alpha+1)(1-r+{|x_-|\over \sqrt{2}})}\big)+{4n^{1\over 2}\over(1-r+{|x_-|\over \sqrt{2}})}
\big({\beta_2\over \alpha+1}+{\beta_3^s\over (\alpha+2)(1-r+{|x_-|\over \sqrt{2}})}\big)\nonumber
\end{eqnarray}
\begin{equation}
+{2n^{1\over 2}\over ({|v_-|\over 2\sqrt{2}}-r)}
\big({\beta_2\over \alpha}+{\beta_3^s\over (\alpha+1)(1-r+{|x_-|\over \sqrt{2}})}\big)
\Big],
\label{4.29}
\end{equation}
for $t\ge 0$.
Then estimate \eqref{lc3} follows from \eqref{4.14}, \eqref{4.15} and \eqref{4.18}, and estimate \eqref{lc4} follows from \eqref{4.26}, \eqref{4.27} and \eqref{4.29}.

Now we prove \eqref{lc7}.
From \eqref{5.18aa}, \eqref{4.9e}, \eqref{02b}, \eqref{4.9} and \eqref{4.9b} it follows that
\begin{eqnarray}
&&|\tilde k(v_-,x_-,f)-v_-|
\le 2n^{1\over 2}\sqrt{1-{|v_-|^2\over c^2}}\int_{-\infty}^{+\infty}|F(z_-+f)(\tau)|d\tau\nonumber\\
&\le&{8n^{3\over 2}\sqrt{1-{|v_-|^2\over c^2}}\over ({|v_-|\over 2\sqrt{2}}-r) (1+{|x_-|\over \sqrt{2}}-r)^\alpha}\big({\beta_1^l\over \alpha}+{\beta_2\over (\alpha+1)(1+{|x_-|\over \sqrt{2}}-r)}\big).\label{5.18b}
\end{eqnarray}
\end{proof}

\section{Proof of Theorem \ref{thm_y}}
\label{sec_thm_y}
Estimate \eqref{t1a} follows from estimate \eqref{2.14}.
The proof of \eqref{t1c} is similar to the proof of \eqref{lc7} given at the end of Section \ref{sec_cont}. 
In this Section we shorten $z_-(v_-,.)$, $a_{sc}(v_-,x_-)$, $a(v_-,x_-)$ and $b_{sc}(v_-,x_-)$ to $z_-$, $a_{sc}$, $a$ and $b_{sc}$.

We prove \eqref{t3}.
First note that
\begin{equation}
a_{sc}-\sqrt{1-{|v_-|^2\over c^2}}\int_{-\infty}^{+\infty}F(.v_-+x_-)(\tau)d\tau
=\Delta_{1,1}+\Delta_{1,2},\label{3.4aa}
\end{equation}
where
\begin{equation}
\Delta_{1,1}:=\lim_{t\to +\infty}\big(\dot A(y_-)-\dot A(0)\big)(t),\label{3.2a}
\end{equation}
\begin{eqnarray}
\Delta_{1,2}&:=&
\left(<\nabla g_j(g^{-1}(v_-)),\int_{-\infty}^{+\infty}\big(F(z_-+x_-)(\tau)-F(. v_-+x_-)(\tau)\big)d\tau >\right)_{j=1\ldots n}\nonumber\\
&&+\left(\int_0^1<\nabla g_j\big(g^{-1}(v_-)+\ep \int_{-\infty}^{+\infty}F(z_-+x_-)(\tau)d\tau\big)-\nabla g_j(g^{-1}(v_-)),\right.\nonumber\\
&&\left.\int_{-\infty}^{+\infty}F(z_-+x_-)(\tau)d\tau> d\ep \right)_{j=1\ldots n},
\label{3.4a}
\end{eqnarray}
and where $<.,.>$ denotes the scalar product in $\R^n$.
For the decomposition \eqref{3.4aa} we used \eqref{02d} and the identity $<v_-,\int_{-\infty}^{+\infty}F(. v_-+x_-)(\tau)d\tau>=0$ to obtain
\begin{equation}
\Big(<\nabla g_j(g^{-1}(v_-)),\int_{-\infty}^{+\infty}F(. v_-+x_-)(\tau)d\tau>\Big)_{j=1\ldots n}
=\sqrt{1-{|v_-|^2\over c^2}}\int_{-\infty}^{+\infty}F(. v_-+x_-)(\tau)d\tau.\label{3.3}
\end{equation} 
We prove the estimates \eqref{t1db}  and  \eqref{3.10a} given below that provide a bound for $\Delta_{1,i}$, $i=1,2$. Then adding those bounds and using the decomposition \eqref{3.4aa} we obtain \eqref{t3}.

We use \eqref{2.29}, and we have
\begin{equation}
|\Delta_{1,1}|
\le{8n^{3\over 2}\max(\beta_1^l,\beta_2,\beta_3^s)
\Big({1-{|v_-|^2\over c^2}}\Big)^{1\over 2}({1\over c}+2n^{1\over 2})\big(1+{1\over {|v_-|\over 2\sqrt{2}}-r}\big)\|y_-\|\over \alpha({|v_-|\over 2\sqrt{2}}-r)(1-r)^{\alpha+2}}.\label{3.2}
\end{equation}
Then $\|y_-\|=\|A(y_-)\|$ is bounded by the right hand side of \eqref{lc1}, and combining this upper bound with \eqref{3.2} we obtain 
\begin{equation}
|\Delta_{1,1}|\le {32n^3\max(\beta_1^l,\beta_2,\beta_3^s)^2\big(1-{|v_-|^2\over c^2}\big)
\big({r\over c}+2n^{1\over 2}(|x_-|+r)+2\big)({1\over c}+2n^{1\over 2})\big(1+{1\over {|v_-|\over 2\sqrt{2}}-r}\big)^2\over \alpha^2({|v_-|\over 2\sqrt{2}}-r)^2(1-r)^{2\alpha+3}}
.\label{t1db}
\end{equation}

We use \eqref{02a},  \eqref{02c} and \eqref{2.9c}, and we obtain
\begin{eqnarray}
|\Delta_{1,2}|
&\le&n^{1\over 2}\sqrt{1-{|v_-|^2\over c^2}}\int_{-\infty}^{+\infty}\big|F(z_-+x_-)(\tau)-F( .v_-+x_-)(\tau)\big|d\tau\nonumber\\
&&+{6n\over c}(1-{|v_-|^2\over c^2})
\Big|\int_{-\infty}^{+\infty}F(z_-+x_-)(\tau)d\tau\Big|^2.\label{3.6}
\end{eqnarray}
Using \eqref{2.8b} and \eqref{2.9b} (with ``$(f,r)$"$=(0,0)$) we obtain
\begin{equation}
\Big|\int_{-\infty}^{+\infty}F(z_-+x_-)(\tau)d\tau\Big|
\le{2^{7\over 2}n\over |v_-|(1+{|x_-|\over \sqrt{2}})^\alpha}\big({\beta_1^l\over \alpha}+{\beta_2\over (\alpha+1)(1+{|x_-|\over \sqrt{2}})}\big).\label{3.7}
\end{equation}
Then we use \eqref{03c} and \eqref{5.3b} (with ``$(w,v,x,q)$"$=(v_-,v_-,0,0)$), and we have 
\begin{eqnarray}
&&\int_{-\infty}^{+\infty}\big|F^l(z_-+x_-)(\tau)-F^l( .v_-+x_-)(\tau)\big|d\tau\nonumber\\
&\le&\int_{-\infty}^{+\infty}{{n\beta_1^l\over c}\sup_{\R}|\dot z_--v_-|d\tau\over \big(1+{|x_-|\over \sqrt{2}}+|\tau|{|v_-|\over 2^{3\over 2}}\big)^{\alpha+1}}
+\int_{-\infty}^{+\infty}{2n^{3\over 2}\beta_2|\tau|\sup_{\R}|\dot z_--v_-|d\tau\over \big(1+{|x_-|\over \sqrt{2}}+|\tau|{|v_-|\over 2^{3\over 2}}\big)^{\alpha+2}}\nonumber\\
&\le&{2^7n^{5\over 2}\beta_1^l\sqrt{1-{|v_-|^2\over c^2}}\over \alpha^2|v_-|^2(1+{|x_-|\over \sqrt{2}})^\alpha}\big({\beta_1^l\over c}
+{2^{5\over 2}n^{1\over 2}\beta_2\over |v_-|}\big).\label{3.9}
\end{eqnarray}
Similarly we use \eqref{03d} and \eqref{5.3b}, and we obtain
\begin{equation}
\hskip -2cm\int_{-\infty}^{+\infty}\big|F^s(z_-+x_-)-F^s( .v_-+x_-)\big|(\tau)d\tau
\le
{2^7n^{5\over 2}\beta_1^l(1-{|v_-|^2\over c^2})^{1\over 2}\big({\beta_2\over c}
+{2^{5\over 2}n^{1\over 2}\beta_3^s\over |v_-|}\big)\over \alpha(\alpha+1)|v_-|^2(1+{|x_-|\over \sqrt{2}})^{\alpha+1}}.\label{3.9b}
\end{equation}
Combining \eqref{3.6}, \eqref{3.7}, \eqref{3.9} and \eqref{3.9b}, we obtain
\begin{eqnarray}
|\Delta_{1,2}|&\le&{2^8\cdot 3n^3(1-{|v_-|^2\over c^2})\over c|v_-|^2(1+{|x_-|\over \sqrt{2}})^{2\alpha}}\big({\beta_1^l\over \alpha}+{\beta_2\over (\alpha+1)(1+{|x_-|\over \sqrt{2}})}\big)^2\nonumber\\
&&+{2^7n^3\beta_1^l\big(1-{|v_-|^2\over c^2}\big)\over \alpha |v_-|^2(1+{|x_-|\over \sqrt{2}})^\alpha}\Big({1\over c}\big({\beta_1^l\over \alpha}+{\beta_2\over (\alpha+1)(1+{|x_-|\over \sqrt{2}})}\big)\nonumber\\
&&+{2^{5\over 2}n^{1\over 2}\over |v_-|}\big({\beta_2 \over \alpha}+{\beta_3^s\over (\alpha+1)(1+{|x_-|\over \sqrt{2}})}\big)\Big).\label{3.10a}
\end{eqnarray}

Now we prove \eqref{t1d} and \eqref{t1e}.
We rewrite $y_+$ as follows 
\begin{equation}
y_+=h_0+h_1,\label{3.14}
\end{equation}
where
\begin{eqnarray}
&&h_0(t):=\int_t^{+\infty}\Big(<\int_0^1 \nabla g_j\big(g^{-1}(a)-\ep\int_\sigma^{+\infty}F(z_-+x_-+y_-)(\tau)d\tau
\label{3.31a}\\
&&-(1-\ep)\int_\sigma^{+\infty} F^l(z_+(a,.))(\tau)d\tau\big)\Big)d\tau\big)d\ep,
\int_\sigma^{+\infty}F^s(z_-+x_-+y_-)(\tau)d\tau>\Big)_{j=1\ldots n}d\sigma,\nonumber
\end{eqnarray}
\begin{eqnarray}
&&h_1(t) 
:=\int_t^{+\infty}\Big(<\int_0^1 \nabla g_j\big(g^{-1}(a)-\ep\int_\sigma^{+\infty}F(z_-+x_-+y_-)(\tau)d\tau
\nonumber\\
&&-(1-\ep)\int_\sigma^{+\infty} F^l(z_+(a,.))(\tau)d\tau\big)\Big)d\tau\big)d\ep,\nonumber\\
&&\int_\sigma^{+\infty}\big(F^l(z_-+x_-+y_-)
-F^l(z_+(a,.))\big)(\tau)d\tau>\Big)_{j=1\ldots n}d\sigma,\label{3.31b}
\end{eqnarray}
for $t\ge 0$.
We estimate $\dot h_0$. 
We also need the following estimate \eqref{3.36}.
For $\ep,\ \ep'\in (0,1)$ and $\tau \ge 0$ we have 
\begin{eqnarray}
&&|(1-\ep)(z_-(\tau)+y_-(\tau))+\ep z_+(a,\tau)+\ep'x_-|\nonumber\\
&\ge&|\ep'x_-+\tau v_-|-(1-\ep)|z_-(\tau)-\tau v_-|-r\tau-r\nonumber\\
&&-\ep|z_+(a,\tau)-\tau a|-\ep |a_{sc}|\tau\nonumber\\
&\ge&\ep'{|x_-|\over  \sqrt{2}}-r+\tau\big({|v_-|\over \sqrt{2}}-r-{2^{9\over 2}n^{3\over 2}\beta_1^l\sqrt{1-{|v_-|^2\over c^2}}\over \alpha |v_-|}\big)-\ep |a_{sc}|\tau.\label{3.33}
\end{eqnarray}
We used \eqref{5.3b} for ``$(v,w,x,q)$"$=(v_-,v_-,0,0)$ and for ``$(v,w,x,q)$"$=(a,a,0,0)$, and we used the identity $|v_-|=a$ that holds by conservation of energy.
Then we use \eqref{t1h} and \eqref{3.33}, and we have for $(\ep,\ep')\in (0,1)^2$ and $\tau \ge 0$
\begin{equation}
|(1-\ep)(z_-(\tau)+y_-(\tau))+\ep z_+(a,\tau)+\ep'x_-|
\ge\ep'{|x_-|\over  \sqrt{2}}-r+\tau\big({|v_-|\over 2\sqrt{2}}-r\big).\label{3.36}
\end{equation}
Then the following estimate that is similar to \eqref{2.9c} holds 
\begin{eqnarray}
&&|g^{-1}(a)-\eta_1\int_{-\infty}^{+\infty}F(z_-+x_-+y_-)(\tau) d\tau-\int_\sigma^{+\infty}\big(\eta_2F(z_-+x_-+y_-)\nonumber\\
&&+\eta_3F^l(z_+(a,.))+\eta_4 F^l(z_+(a,.)+x_-)\big)(\tau)d\tau|\ge  {1\over 2}|g^{-1}(v_-)|,\label{B9}
\end{eqnarray}
for $(\sigma,\eta_1,\eta_2,\eta_3,\eta_4)\in [0,+\infty)\times [0,1]^4$, $\eta_1+\eta_2+\eta_3+\eta_4\le 1$.

From \eqref{B9}, \eqref{02a}, \eqref{03b} and \eqref{3.31a} it follows that
\begin{eqnarray}
|\dot h_0( t)|
&\le&2n^{1\over 2}\big(1- {|v_-|^2\over c^2}\big)^{1\over 2}\int_t^{+\infty} |F^s(z_-+x_-+y_-)|(\tau)d\tau \nonumber\\
&\le &{4\beta_2n^{3\over 2}\sqrt{1- {|v_-|^2\over c^2}}\over (\alpha+1)({|v_-|\over 2^{3\over 2}}-r)
 (1+{|x_-|\over \sqrt{2}}-r+({|v_-|\over 2^{3\over 2}}-r)t)^{\alpha+1}},\label{3.45a}
\end{eqnarray}
for $t\ge 0$.
Now set
\begin{equation} 
\delta_r:=
\max\Big(\sup_{(0,+\infty)}|b_{sc}+y_+|,\sup_{(0,+\infty)}(1-r+({|v_-|\over 2^{3\over 2}}-r)s)^{-1}|\dot y_+(s)|\Big).\label{3.40}
\end{equation}
We remind that $\delta_r$ is finite by Lemma \ref{lem:comp} (for ``$(x,z)$"$=(z_-+x_-+y_-, z_+(a,.))$).
Then we use \eqref{02a}, \eqref{03c}, \eqref{B9} and \eqref{3.36}, and we obtain
\begin{eqnarray}
|\dot h_1(t)|
&\le &2n^{1\over 2}\sqrt{1-{|v_-|^2\over c^2}}\int_t^{+\infty}{2n^{3\over 2}\beta_2|x_-|+({n\beta_1^l\over c}+2n^{3\over 2}\beta_2)\delta_r\over (1-r+\tau\big({|v_-|\over 2\sqrt{2}}-r\big))^{\alpha+2}}d\tau d\sigma\nonumber\\
&\le&{2n^{3\over 2}\sqrt{1-{|v_-|^2\over c^2}}\big(2n^{1\over 2}\beta_2|x_-|+({\beta_1^l\over c}+2n^{1\over 2}\beta_2)\delta_r\big)\over (\alpha+1)\big({|v_-|\over 2\sqrt{2}}-r\big)(1-r+t\big({|v_-|\over 2\sqrt{2}}-r\big))^{\alpha+1}},\label{3.45b}
\end{eqnarray}
for $t\ge 0$.
Hence combining \eqref{3.14}, \eqref{3.45a} and \eqref{3.45b} we have 
\begin{equation}
|\dot y_+(t)|
\le {2n^{3\over 2}\big(2\beta_2+2n^{1\over 2}\beta_2|x_-|+({\beta_1^l\over c}+2n^{1\over 2}\beta_2)\delta_r\big)\sqrt{1-{|v_-|^2\over c^2}}\over (\alpha+1)\big({|v_-|\over 2\sqrt{2}}-r\big)\big(1-r+t({|v_-|\over 2^{3\over 2}}-r)\big)^{\alpha+1}},\label{3.46}
\end{equation}
for $t\ge 0$.
In addition from \eqref{3.13a} and \eqref{2.15} at $t=0$ it follows that
\begin{equation}
|b_{sc}+y_+(0)|=|A(y_-)(0)|\le{2n^{3\over 2}\left({1-{|v_-|^2\over c^2}}\right)^{1\over 2}\big({r\beta_1^l\over c}+2\beta_2(n^{1\over 2}(|x_-|+r)+1)\big)\over \alpha(\alpha+1)\big({|v_-|\over 2\sqrt{2}}-r\big)^2(1-r)^\alpha}.\label{B30}
\end{equation}
From \eqref{3.40}, \eqref{3.46} and \eqref{B30} it follows that
\begin{equation}
\delta_r\le{2n^{3\over 2}\sqrt{1- {|v_-|^2\over c^2}}\big({r\beta_1^l\over c}+2\beta_2(n^{1\over 2}(2|x_-|+r)+2)+({\beta_1^l\over c}+2n^{1\over 2}\beta_2)\delta_r\big)\max\big(1,{1\over \alpha({|v_-|\over 2^{3\over 2}}-r)}\big)\over (\alpha+1)({|v_-|\over 2^{3\over 2}}-r)(1-r)^\alpha}.\label{3.45}
\end{equation}
Then we use condition \eqref{t1h} and we obtain
\begin{equation}
{\delta_r\over 2}\le{2n^{3\over 2}\sqrt{1- {|v_-|^2\over c^2}}\big({r\beta_1^l\over c}+2\beta_2(n^{1\over 2}(2|x_-|+r)+2)\big)\over (\alpha+1)({|v_-|\over 2^{3\over 2}}-r)(1-r)^\alpha}\max\big(1,{1\over \alpha({|v_-|\over 2^{3\over 2}}-r)}\big).\label{B10}
\end{equation}
Estimate \eqref{B10} and condition \eqref{t1h} also provide the following estimate
\begin{equation}
({\beta_1^l\over c}+2n^{1\over 2}\beta_2)\delta_r\le {r\beta_1^l\over c}+2\beta_2(n^{1\over 2}(2|x_-|+r)+2).\label{B11}
\end{equation}
Estimate \eqref{t1e} follows from \eqref{3.46} and \eqref{B11}. 
Then we integrate over $(0,+\infty)$ both sides of \eqref{t1e} and we obtain a bound on $|y_+(0)|$. Then we add this bound with the bound given in \eqref{2.15} for $A(y_-)(0)$, and we use \eqref{3.13a} and we obtain \eqref{t1d}.

It remains to prove \eqref{t4}.
From \eqref{3.13a}, \eqref{3.14}, \eqref{3.31a} and \eqref{3.31b} at $t=0$ it follows by straightforward computations that 
\begin{eqnarray}
&& \sum_{j=1}^8 \Delta_{2,j}=b_{sc}(v_-,x_-)-\int_{-\infty}^0\Big(g\big(g^{-1}(v_-)+\int_{-\infty}^\sigma F^l(z_-+x_-)(\tau)d\tau\big)\nonumber\\
&&-g\big(g^{-1}(v_-)+\int_{-\infty}^\sigma F^l(z_-)(\tau)d\tau\big)\Big)d\sigma
-\int_0^{+\infty}\Big(g\big(g^{-1}(a)-\int_\sigma^{+\infty}F^l(z_+(a,.)+x_-)(\tau)d\tau\big)\nonumber\\
&&-g\big(g^{-1}(a)-\int_\sigma^{+\infty}F^l(z_+(a,.))(\tau)d\tau\big)\Big)d\sigma-\Big(<\nabla g_j(g^{-1}(v_-)),\nonumber\\
&&\int_{-\infty}^0\int_{-\infty}^\sigma F^s(. v_-+x_-)(\tau)d\tau d\sigma-\int_0^{+\infty}\int_\sigma^{+\infty} F^s(. v_-+x_-)(\tau)d\tau d\sigma>\Big)_{j=1\ldots n},\label{3.90}
\end{eqnarray}
where
\begin{equation}
\Delta_{2,1}:=A(y_-)(0)-A(0)(0),\label{B1}
\end{equation}
\begin{eqnarray}
&&\Delta_{2,2}:=
\int_{-\infty}^0\Big(<\nabla g_j\big(g^{-1}(v_-)+\ep\int_{-\infty}^\sigma F(z_-+x_-)(\tau)d\tau
\label{B2}\\
&&+(1-\ep)\int_{-\infty}^\sigma F^l(z_-+x_-)(\tau)d\tau\big)-\nabla g_j\big(g^{-1}(v_-)),
\int_{-\infty}^\sigma F^s(z_-+x_-)(\tau)d\tau>\Big)_{j=1\ldots n}d\sigma,\nonumber
\end{eqnarray}
\begin{equation}
\Delta_{2,3}:=\int_{-\infty}^0\Big(<\nabla g_j(g^{-1}(v_-)),\int_{-\infty}^\sigma(F^s(z_-+x_-)
-F^s(. v_-+x_-))(\tau)d\tau>\Big)_{j=1\ldots n}d\sigma,\label{B3}
\end{equation}
\begin{eqnarray}
\Delta_{2,4}&:=&-
\int_0^{+\infty}\Big(<\int_0^1\big(\nabla g_j\big(g^{-1}(a)-\ep\int_\sigma^{+\infty}F(z_-+x_-+y_-)(\tau)d\tau\nonumber\\
&&
-(1-\ep)\int_\sigma^{+\infty} F^l(z_+(a,.))(\tau)d\tau\big)-\nabla g_j(g^{-1}(v_-))\Big)d\ep,\nonumber\\
&&\int_\sigma^{+\infty} F^s(z_-+x_-+y_-)(\tau)d\tau>\Big)_{j=1\ldots n}d\sigma,\label{B4}
\end{eqnarray}
\begin{equation}
\Delta_{2,5}:=-\int_0^{+\infty}\Big(<\nabla g_j(g^{-1}(v_-)),\int_\sigma^{+\infty}(F^s(z_-+x_-+y_-)
-F^s(z_-+x_-))(\tau)d\tau>\Big)_{j=1\ldots n}d\sigma,\label{B5}
\end{equation}
\begin{equation}
\Delta_{2,6}:=-\int_0^{+\infty}\Big(<\nabla g_j(g^{-1}(v_-)),\int_\sigma^{+\infty}(F^s(z_-+x_-)
-F^s(. v_-+x_-))(\tau)d\tau>\Big)_{j=1\ldots n}d\sigma,\label{B6}
\end{equation}
\begin{eqnarray}
\Delta_{2,7}&:=&-\int_0^{+\infty}\Big(<\int_0^1 \nabla g_j\big(g^{-1}(a)-\ep\int_\sigma^{+\infty}F(z_-+x_-+y_-)(\tau)d\tau
\nonumber\\
&&-(1-\ep)\int_\sigma^{+\infty} F^l(z_+(a,.))(\tau)d\tau\big)d\ep,\label{B7}\\
&&\int_\sigma^{+\infty}\big(F^l(z_-+x_-+y_-)
-F^l(z_+(a,.)+x_-)\big)(\tau)d\tau>\Big)_{j=1\ldots n}d\sigma,\nonumber
\end{eqnarray}
\begin{eqnarray}
&&\Delta_{2,8}:=-\int_0^{+\infty}\Big(<\int_0^1 \Big(\nabla g_j\big(g^{-1}(a)-\ep\int_\sigma^{+\infty}F(z_-+x_-+y_-)(\tau)d\tau
\nonumber\\
&&-(1-\ep)\int_\sigma^{+\infty} F^l(z_+(a,.))(\tau)d\tau\big)\label{B8}\\
&&-\nabla g_j\big(g^{-1}(a)-\ep\int_\sigma^{+\infty}F^l(z_+(a,.)+x_-)(\tau)d\tau
-(1-\ep)\int_\sigma^{+\infty} F^l(z_+(a,.))(\tau)d\tau\big)\Big)d\ep,\nonumber\\
&&\int_\sigma^{+\infty}\big(F^l(z_+(a,.)+x_-)
-F^l(z_+(a,.))\big)(\tau)d\tau>\Big)_{j=1\ldots n}d\sigma.\nonumber
\end{eqnarray}
Using \eqref{2.27} (for $(f_1,f_2)=(y_-,0)$ at time $t=0$) and \eqref{lc1} we obtain
\begin{equation}
|\Delta_{2,1}|
\le{12n^3\big({1-{|v_-|^2\over c^2}}\big)\max(\beta_1^l,\beta_2,\beta_3^s)^2\big({r\over c}+2(n^{1\over 2}(|x_-|+r)+1)\big)({1\over c}
+2n^{1\over 2}))\big(1+{1\over {|v_-|\over 2\sqrt{2}}-r}\big)\over \alpha^2(\alpha+1)({|v_-|\over 2\sqrt{2}}-r)^3(1-r)^{2\alpha+2}}
.\label{3.15}
\end{equation}
Using \eqref{02c}, \eqref{2.9c} and then \eqref{03a}, \eqref{03b} and \eqref{2.5b} (with  ``$(f,r)$"$=(0,0)$) we have
\begin{eqnarray}
&&|\Delta_{2,2}|
\le {6n\over c}(1-{|v_-|^2\over c^2})
\int_{-\infty}^0\max\Big(\int_{-\infty}^\sigma|F(z_-(v_-,.)+x_-)(\tau)|d\tau,\nonumber\\
&&\int_{-\infty}^\sigma|F^l(z_-(v_-,.)+x_-)(\tau)|d\tau\Big)
\int_{-\infty}^\sigma|F^s(z_-(v_-,.)+x_-)(\tau)|d\tau d\sigma\nonumber\\
&&\le 
\int_{-\infty}^0\Big({2^{5\over 2}\beta_1^l n\over \alpha |v_-|(1+{|x_-|\over \sqrt{2}})^\alpha}+{2^{5\over 2}\beta_2 n\over (\alpha+1) |v_-|(1+{|x_-|\over \sqrt{2}})^{\alpha+1}}\Big)
{{6n^2\over c}(1-{|v_-|^2\over c^2})2^{5\over 2}\beta_2 d\sigma\over (\alpha+1)|v_-|(1+{|x_-|\over \sqrt{2}}+{|v_-|\over 2\sqrt{2}}|\sigma|)^{\alpha+1}}\nonumber\\
&&\le 
{48n^3(1-{|v_-|^2\over c^2})\max(\beta_1^l,\beta_2)^2\over c\alpha^2 (\alpha+1)({|v_-|\over 2^{3\over 2}})^3(1+{|x_-|\over \sqrt{2}})^{2\alpha}}.\label{3.23}
\end{eqnarray}
We use \eqref{02a}, \eqref{03d}, \eqref{2.5b} (with ``$(f,r)$"$=(0,0)$) and \eqref{5.3b}, and we obtain 
\begin{eqnarray}
\max(|\Delta_{2,3}|,|\Delta_{2,6}|)
&\le&{2^{9\over 2}n^2\beta_1^l(1-{|v_-|^2\over c^2})\over \alpha |v_-|}\int_0^{+\infty}\int_\sigma^{+\infty}\big[{n\beta_2\over c}(1+{|x_-|\over \sqrt{2}}+{|v_-|\over 2\sqrt{2}}\tau)^{-\alpha-2}\nonumber\\
&&+2n^{3\over 2}\beta_3^s\tau(1+{|x_-|\over \sqrt{2}}+({|v_-|\over 2\sqrt{2}}\tau)^{-\alpha-3}\big]d\tau d\sigma\nonumber\\
&\le&{8n^3\max(\beta_1^l,\beta_2,\beta_3^s)^2(1-{|v_-|^2\over c^2}) \over \alpha^2(\alpha+1)({|v_-|\over 2^{3\over 2}})^3(1+{|x_-|\over \sqrt{2}})^\alpha}
\Big({1\over c}+{2n^{1\over 2}\over{|v_-|\over 2^{3\over 2}}}\Big).\label{3.18}
\end{eqnarray}
We use \eqref{02c} and \eqref{B9}, and then we use \eqref{03a}, \eqref{03b} and \eqref{3.36}, and we have
\begin{eqnarray}
|\Delta_{2,4}|
&\le&{6n\over c}(1-{|v_-|^2\over c^2})
\int_0^{+\infty}\Big(\int_{-\infty}^{+\infty}|F(z_-+x_-+y_-)(\tau)|d\tau\nonumber\\
&&+\int_\sigma^{+\infty}|F^l(z_+(a,.))(\tau)|d\tau\Big)\int_\sigma^{+\infty}|F^s(z_-+x_-+y_-)(s)|ds d\sigma\nonumber\\
&\le&{6n\over c}(1-{|v_-|^2\over c^2})
\int_0^{+\infty}\Big(\int_0^{+\infty}{4n\beta_2 d\tau\over (1+{|x_-|\over \sqrt{2}}-r+({|v_-|\over 2\sqrt{2}}-r)\tau)^{\alpha+2}}\nonumber\\
&&+\int_0^{+\infty}{6n\beta_1^l d\tau\over (1-r+({|v_-|\over 2\sqrt{2}}-r)\tau)^{\alpha+1}}\Big)
\int_\sigma^{+\infty}{2n\beta_2 ds d\sigma\over (1+{|x_-|\over \sqrt{2}}-r+({|v_-|\over 2\sqrt{2}}-r)s)^{\alpha+2}}\nonumber\\
&\le&{24n^3(1-{|v_-|^2\over c^2})\max(\beta_1^l,\beta_2)^2\over c\alpha^2(\alpha+1)({|v_-|\over 2\sqrt{2}}-r)^3(1-r)^{2\alpha}}\Big(3+{2\over
1+{|x_-|\over \sqrt{2}}-r}\Big).
\label{3.16}
\end{eqnarray}
We use \eqref{02a}, \eqref{2.9ab} and \eqref{lc1} ($y_-=A(y_-)$), and we obtain 
\begin{eqnarray}
&&|\Delta_{2,5}|
\le n^{1\over 2}\sqrt{1-{|v_-|^2\over c^2}}\int_0^{+\infty}\int_\sigma^{+\infty}|F^s(z_-+x_-+y_-)(\tau)
-F^s(z_-+x_-)(\tau)|d\tau d\sigma\nonumber\\
&&\le n^{1\over 2}\sqrt{1-{|v_-|^2\over c^2}}\|y_-\|\int_0^{+\infty}\int_\sigma^{+\infty}\Big({{n\beta_2\over c}\over(1+{|x_-|\over \sqrt{2}}-r+({|v_-|\over 2\sqrt{2}}-r)\tau)^{\alpha+2}}\nonumber\\
&&+{2n^{3\over 2}\beta_3^s(1+\tau)\over(1+{|x_-|\over \sqrt{2}}-r+({|v_-|\over 2\sqrt{2}}-r)\tau)^{\alpha+3}}\Big)d\tau d\sigma\nonumber
\end{eqnarray}
\begin{equation}
\le{4n^3(1-{|v_-|^2\over c^2})\max(\beta_1^l,\beta_2,\beta_3^s)^2 \big({r\over c}+2(n^{1\over 2}(|x_-|+r)+1)\big)
\big({1\over c}+{2n^{1\over 2}\over {|v_-|\over 2\sqrt{2}}-r}+n^{1\over 2}\Big)\big(1+{1\over {|v_-|\over 2\sqrt{2}}-r}\big)\over \alpha^2(\alpha+1)({|v_-|\over 2\sqrt{2}}-r)^3(1-r)^{2\alpha+2}}
.\label{3.17}
\end{equation}
From \eqref{B9}, \eqref{02a}, \eqref{03c}, \eqref{3.36} and \eqref{B10} it follows that
\begin{equation*}
|\Delta_{2,7}|\le 2n^{1\over 2}\sqrt{1-{|v_-|^2\over c^2}}\int_0^{+\infty}\int_\sigma^{+\infty}{({n\beta_1^l\over c}+2n^{3\over 2}\beta_2)\delta_r
d\tau d\sigma\over (1+{|x_-|\over \sqrt{2}}-r+({|v_-|\over 2^{3\over 2}}-r)\tau)^{\alpha+2}}
\end{equation*}
\begin{equation}
\le{8n^3(1- {|v_-|^2\over c^2})({\beta_1^l\over c}+2n^{1\over 2}\beta_2)\big({r\beta_1^l\over c}+2\beta_2(n^{1\over 2}(2|x_-|+r)+2)\big)\max\big(1,{1\over \alpha({|v_-|\over 2^{3\over 2}}-r)}\big)
\over \alpha(\alpha+1)^2({|v_-|\over 2^{3\over 2}}-r)^3(1-r)^{2\alpha}}.
\label{B12}
\end{equation}
Combining \eqref{3.36}, \eqref{02c}, \eqref{B9}, \eqref{03c}, \eqref{03a} and \eqref{03b} we obtain
\begin{eqnarray}
|\Delta_{2,8}|&\le&{6n(1-{|v_-|^2\over c^2})\over c}\int_0^{+\infty}\Big(\int_\sigma^{+\infty} |F(z_-+x_-+y_-)(\tau)|
\nonumber\\
&&+|F^l(z_+(a,.)+x_-)(\tau)|d\tau\Big) \int_\sigma^{+\infty}{2n^{3\over2} \beta_2|x_-|\over (1+{|v_-|\over 2^{3\over 2}}\tau)^{\alpha+2} }d\tau d\sigma\nonumber\\
&\le&{24n^{7\over 2}\beta_2|x_-|(1-{|v_-|^2\over c^2})
\big({2\beta_1^l\over \alpha}+{\beta_2^s\over (\alpha+1)(1+{|x_-|\over \sqrt{2}}-r)}\big)\over c\alpha(\alpha+1)\big({|v_-|\over 2^{3\over 2}}\big)^2({|v_-|\over 2^{3\over 2}}-r)(1+{|x_-|\over \sqrt{2}}-r)^\alpha}.\label{3.38}
\end{eqnarray}
Then we add the bounds on the right-hand sides of \eqref{3.15}--\eqref{3.38}, and we use \eqref{3.90}, and we obtain \eqref{t4}.
\hfill$\Box$

\section{Proof of Lemma \ref{lem_cont3} and Theorem \ref{thm_y2}}
\label{sec_cont3}
\subsection{Preliminary Lemma}
\begin{lemma}
\label{lem_scatinit2} 
Let $(v,w,x,q)\in \B(0,c)^2\times \R^n\times\B(0,1)$ so that $|v|=|w|\not=0$, $|v-w|<{|v|\over 2^{5\over 2}}$ and $v\cdot x=0$.
Assume that
\begin{equation}
{2^{7\over 2}n^{3\over 2}\big({1\over c}+2n^{1\over 2}\big)\max(\beta_1^l,\beta_2^l)\sqrt{1-{|v|^2\over c^2}}\over \alpha|v|\big(1+{|x|\over \sqrt{2}}-|q|\big)^\alpha}\big(1+{2^{3\over 2}\over |v| }\big)\le 1.\label{5.3g}
\end{equation}
Then the following estimates are valid
\begin{equation}
|\eta_1(z_-(v,x,t)+f(t))+\eta_2 z_+(w,x+q,t)+\eta_3 z_+(w,x+q',t)|\ge {|x|\over \sqrt{2}}-(\eta_2+\eta_3)|q|-\eta_1 r+({|v|\over 2^{3\over 2}}-\eta_1 r)t,\label{5.3d}
\end{equation}
\begin{equation}
\delta_{+,q,q'}:=\max\big(\sup_{(0,+\infty)}|\omega_{+,q,q'}|,\sup_{t\in(0,+\infty)}(1+{|x|\over \sqrt{2}}-|q|+t{|v|\over 2^{3\over 2}})
|\dot \omega_{+,q,q'}(t)|\big)\le 2|q-q'|,\label{5.3e}
\end{equation}
for $(r,\eta_1,\eta_2,\eta_3,t,q')\in (0,\min({|v|\over 2^{3\over 2}},1))\times [0,1]^3\times [0,+\infty)\times\R^n$ and for  $f\in M_{r,v,x}$ so that $\eta_1+\eta_2+\eta_3=1$, $|q'|\le |q|$, where $\omega_{+,q,q'}=z_+(w,x+q,.)-z_+(w,x+q',.)$. 
\end{lemma}

\begin{proof}[Proof of Lemma \ref{lem_scatinit2}]
We have
\begin{eqnarray}
&&|\eta_1(z_-(v,x,t)+f(t))+\eta_2 z_+(w,x+q,t)+\eta_3 z_+(w,x+q',t)|\nonumber\\
&&\ge |x+\eta_2 q+\eta_3 q'+tv|-\eta_1|z_-(v,x,t)-tv-x|-\eta_1 |f(t)|-(\eta_2+\eta_3)|v-w||t|\nonumber\\
&&-\eta_2|z_+(w,x+q,t)-wt-x-q|-\eta_3|z_+(w,x+q',t)-wt-x-q'|.\label{5.13}
\end{eqnarray}
Then we use \eqref{5.3b}, \eqref{5.1a}, \eqref{4.5b}, the equality $x\cdot v=0$ and the estimates $|q'|\le |q|$ and $|v-w|\le{|v|\over 2^{5\over 2}}$ to obtain \eqref{5.3d}.

From Lemma \ref{lem:comp} it follows that $\delta_{+,q,q'}$ is finite. 
Combining \eqref{03c} and \eqref{5.13} we obtain 
\begin{equation}
\big|F^l(z_+(w,x+q,.))(\tau)-F^l(z_+(w,x+q',.))(\tau))\big|
\le{\big({n\beta_1^l\over c}+2n^{3\over 2}\beta_2^l\big)\delta_{+,q,q'}\over \big(1+{|x|\over \sqrt{2}}-|q|+\tau{|v|\over 2^{3\over 2}}\big)^{\alpha+2}},\label{5.15}
\end{equation}
for $\tau\ge 0$.
Note that 
\begin{eqnarray}
\dot \omega_{+,q,q'}(t)&=&g\big(g^{-1}(w)-\int_t^{+\infty}F^l(z_+(w,x+q,.))(\tau)d\tau\big)\nonumber\\
&&-g\big(g^{-1}(w)-\int_t^{+\infty}F^l(z_+(w,x+q',.))(\tau)d\tau\big),\label{5.16}
\end{eqnarray}
for $t\ge 0$. Combining \eqref{5.15}, \eqref{5.16} and \eqref{02b} and \eqref{5.7} we have 
\begin{equation}
|\dot \omega_{+,q,q'}(t)|
\le  {2^{5\over 2}n^{3\over 2}\sqrt{1-{|v|^2\over c^2}}\big({\beta_1^l\over c}+2n^{1\over 2}\beta_2^l\big)\delta_{+,q,q'}\over (\alpha+1)|v|\big(1+{|x|\over \sqrt{2}}-|q|+t{|v|\over 2^{3\over 2}}\big)^{\alpha+1}},\label{5.17a}
\end{equation}
for $t\ge 0$.
Then we use the estimate \eqref{5.17a} and the estimate $|g(t)|\le |g(0)|+\int_0^t|\dot g(s)|ds$ for $t\ge 0$ and $g=\omega_{+,q,q'}$ ($g(0)=q-q'$), and we obtain a bound on  $\omega_{+,q,q'}$, and then we have
\begin{equation}
\delta_{+,q,q'}\le |q-q'|+{2^{5\over 2}n^{3\over 2}\big({\beta_1^l\over c}+2n^{1\over 2}\beta_2^l\big)\max\big(1,{2^{3\over 2}\over \alpha |v| }\big)\sqrt{1-{|v|^2\over c^2}}\delta_{+,q,q'}\over (\alpha+1)|v|\big(1+{|x|\over \sqrt{2}}-|q|\big)^\alpha}. \label{5.17b}
\end{equation} 
Using \eqref{5.3g} we obtain $\delta_{+,q,q'}\le|q-q'|+{\delta_{+,q,q'}\over 2}$,
which proves \eqref{5.3e}.
\end{proof}

For the rest of the text we shorten $z_-(v_-,x_-,.)$, $\tilde a_{sc}(v_-,x_-)$, $\tilde a(v_-,x_-)$ and $\tilde b_{sc}(v_-,x_-)$ to $z_-$, $\tilde a_{sc}$, $\tilde a$ and $\tilde b_{sc}$.

\subsection{Proof of Lemma \ref{lem_cont3}}
We first need to estimate  
\begin{equation}
\delta:=
\max\big(\sup_{(0,+\infty)}|\omega|,\sup_{t\in (0,+\infty)}(1+{|x_-|\over \sqrt{2}}-r+t({|v_-|\over 2^{3\over 2}}-r))|\dot \omega(t)|\big).\label{5.18}
\end{equation}
where $\omega=z_-+y_--z_+(\tilde a,x_-,.)$. Then under condition \eqref{10.1} we have
\begin{equation}
|g^{-1}(v_-)|-\int_{-\infty}^{+\infty}|F(z_-+y_-)(\tau)|d\tau
-\int_0^{+\infty}|F^l(z_+(\tilde a,x_-,.))(\tau)|d\tau
\ge{1\over 2}|g^{-1}(v_-)|.\label{5.23}
\end{equation}
Note also that $\omega(0)=y_-(0)$ and 
\begin{equation}
\dot \omega(t)
=g\big(g^{-1}(\tilde a)-\int_t^{+\infty}F(z_-+y_-)(\tau)d\tau\big)
-g\big(g^{-1}(\tilde a)-\int_t^{+\infty}F^l(z_+(\tilde a,x_-,.))(\tau)d\tau\big),\label{5.24}
\end{equation}
for $t\ge 0$.
Then similarly to \eqref{3.46} we obtain
\begin{equation}
|\dot \omega(t)|\le 
{2n^{3\over 2}\sqrt{1-{|v_-|^2\over c^2}}\big(2\beta_2+\delta({\beta_1^l\over c}+2n^{1\over 2}\beta_2)\big)\over (\alpha+1)({|v_-|\over 2^{3\over 2}}-r)\big(1-r+{|x_-|\over \sqrt{2}}+t\big({|v_-|\over 2^{3\over 2}}-r\big)\big)^{\alpha+1}},\label{5.25a}
\end{equation}
\begin{equation}
\delta\le |y_-(0)|+
{2n^{3\over 2}\sqrt{1-{|v_-|^2\over c^2}}\big(2\beta_2+\delta({\beta_1^l\over c}+2n^{1\over 2}\beta_2)\big)
\max\big(1,{1\over \alpha({|v_-|\over 2^{3\over 2}}-r)}\big)\over (\alpha+1)({|v_-|\over 2^{3\over 2}}-r)\big(1-r+{|x_-|\over \sqrt{2}}\big)^\alpha},\label{5.25c}
\end{equation}
for $t\ge 0$.
Then  from \eqref{10.1} and \eqref{4.15} ($y_-(0)=(\A(y_-)(0)$) and from the estimate $r\le {1\over 2}$ 
it follows that
\begin{equation}
{\delta\over 2}
\le
{2n^{3\over 2}\sqrt{1-{|v_-|^2\over c^2}}\big(2\beta_2+{r\beta_1^l\over c}+(n^{1\over 2}r+1)2\beta_2\big)
\max\big(1,{1\over \alpha({|v_-|\over 2^{3\over 2}}-r)}\big)\over (\alpha+1)({|v_-|\over 2^{3\over 2}}-r)\big({1\over 2}+{|x_-|\over \sqrt{2}}\big)^\alpha}.\label{5.25e}
\end{equation}
In addition under condition \eqref{10.1} we have $\delta({\beta_1^l\over c}+2n^{1\over 2}\beta_2)\le \big(2\beta_2+{r\beta_1^l\over c}+(n^{1\over 2}r+1)2\beta_2\big)$, and from \eqref{5.25a} it follows that
\begin{equation}
|\dot \omega(t)|\le 
{2n^{3\over 2}\sqrt{1-{|v_-|^2\over c^2}}\big({r\beta_1^l\over c}+(n^{1\over 2}r+3)2\beta_2\big)\over (\alpha+1)({|v_-|\over 2^{3\over 2}}-r)\big(1-r+{|x_-|\over \sqrt{2}}+t\big({|v_-|\over 2^{3\over 2}}-r\big)\big)^{\alpha+1}},\label{5.25g}
\end{equation}
Then note that 
\begin{equation}
\G_{v_-,x_-}(q)=\A(y_-)(0)+\int_0^{+\infty}(\dot \omega(s)-\dot \omega_{+,q,0}(s))ds,\label{B20}
\end{equation}
where $\omega_{+,q,q'}$ is defined in Lemma \ref{lem_scatinit2} for ``$w$"$=\tilde a$ and for any $q'\in \overline{\B(0,{1\over 2})}$.
Then we use \eqref{4.15} at $t=0$,  \eqref{5.17a} and \eqref{5.3e} ("$q'=0$") and \eqref{5.25g}, and we obtain
\begin{eqnarray}
|\G_{v_-,x_-}(q)|
&\le& {4n^{3\over 2}\left({1-{|v_-|^2\over c^2}}\right)^{1\over 2}\big(({\beta_1^l\over c}+2n^{1\over 2}\beta_2)r+4\beta_2\big)\over \alpha(\alpha+1)\big({|v_-|\over 2\sqrt{2}}-r\big)^2\big(1-r+{|x_-|\over \sqrt{2}}\big)^\alpha}\nonumber\\
&&+{2^5n^{3\over 2}\sqrt{1-{|v_-|^2\over c^2}}({\beta_1^l\over c}+2n^{1\over 2}\beta_2)|q|\over \alpha(\alpha+1)|v_-|^2(1+{|x_-|\over \sqrt{2}}-|q|)^\alpha},\label{6.9}
\end{eqnarray}
for $|q|\le {1\over 2}$. Estimates \eqref{lc3a} follow from \eqref{6.9}, condition \eqref{10.1} and $\max(r,|q|)\le {1\over 2}$.
Note also that for $(q,q')\in \overline{\B(0, {1\over 2})}^2$ 
\begin{equation}
\G_{v_-,x_-}(q)-\G_{v_-,x_-}(q')=\int_0^{+\infty}\dot \omega_{+,q,q'}(s)ds.
\end{equation}
Then we use \eqref{5.17a} and \eqref{5.3e} and condition \eqref{10.1}, and we obtain \eqref{lc3b}.
\hfill $\Box$

\subsection{Proof of Theorem \ref{thm_y2}}
Estimate \eqref{t2a} follows from the identity $\A(y_-)=y_-$ and \eqref{4.14} at $t=0$.
Estimate \eqref{t2c} follows from \eqref{lc7} and $r\le {1\over 2}$.
Estimate \eqref{t2d} follows from \eqref{lc3a}.
We add \eqref{5.25g} and \eqref{5.17a} for $(q,q')=(\tilde b_{sc}, 0)$, and we use \eqref{5.3e} and  we obtain 
\begin{equation}
|\dot y_+(t)|\le 
{8n^2\sqrt{1-{|v_-|^2\over c^2}}\max(\beta_1^l,\beta_2)({1\over c}+1)(2+|\tilde b_{sc}|)\over (\alpha+1)({|v_-|\over 2^{3\over 2}}-r)\big({1\over 2}+{|x_-|\over \sqrt{2}}+t\big({|v_-|\over 2^{3\over 2}}-r\big)\big)^{\alpha+1}}.\label{D1}
\end{equation}
Then we use $\tilde b_{sc}\le {1\over 2}$, and we obtain \eqref{t2e}.

Now we prove \eqref{t2f}.
Note that
\begin{equation}
\tilde a_{sc}(v_-,x_-)-\tilde W(v_-,x_-)-\sqrt{1-{|v_-|^2\over c^2}}\int_{-\infty}^{+\infty}F^s(\tau v_-+x_-,v_-)d\tau=:\sum_{j=1}^3\Delta_{fc,j},
\end{equation}
where
\begin{eqnarray}
\Delta_{fc,1}
&:=&\tilde a(v_-,x_-)-g\big(g^{-1}(v_-)+\int_{-\infty}^0 F^l(z_-)(\tau)d\tau\nonumber\\
&&+\int_{-\infty}^{+\infty}F^s(z_-)(\tau)d\tau+\int_0^{+\infty}F^l(z_+(\tilde a,x_-,.))(\tau)d\tau\big),\label{5.26a}
\end{eqnarray}
\begin{eqnarray}
\Delta_{fc,2}
&:=&\Big(<\int_0^1\big(\nabla g_j\big(g^{-1}(v_-)+\int_{-\infty}^0 F^l(z_-)(\tau)d\tau+\int_0^{+\infty}F^l(z_+(\tilde a,x_-,.))(\tau)d\tau\nonumber\\
&&\hskip-2cm+\ep \int_{-\infty}^{+\infty}F^s(z_-)(\tau)d\tau\big)-\nabla g_j(g^{-1}(v_-))\big)d\ep,\int_{-\infty}^{+\infty}F^s(z_-)(\tau)d\tau>\Big)_{j=1\ldots n},\label{5.28}
\end{eqnarray}
\begin{equation}
\Delta_{fc,3}:=\Big(<\nabla g_j(g^{-1}(v_-))\big),\int_{-\infty}^{+\infty}(F^s(z_-)(\tau)-F^s(. v_-+x_-)(\tau))d\tau >\Big)_{j=1\ldots n},\label{5.30}
\end{equation}
and where $<.,.>$ denotes the scalar product in $\R^n$.
Under condition \eqref{10.1} it follows that
\begin{eqnarray}
&&|g^{-1}(v_-)|-\int_{-\infty}^{+\infty}\big((|F^l|+|F^s|)(z_-+y_-)(\tau)+(|F^l|+|F^s|)(z_-)(\tau)\big)d\tau\nonumber\\
&&+|F^l(z_+(\tilde a,x_-,.))|(\tau)\big)d\tau
\ge{|g^{-1}(v_-)|\over 2}.\label{D5}
\end{eqnarray}
Then we use \eqref{02b} and we have
\begin{eqnarray}
|\Delta_{fc,1}|
&\le&{2n^{1\over 2}\sqrt{1-{|v_-|^2\over c^2}}}
\Big(\int_{-\infty}^0 |F(z_-+y_-)(\tau)-F(z_-)(\tau)|d\tau\nonumber\\
&&+\int_0^{+\infty} |F^s(z_-+y_-)(\tau)-F^s(z_-)(\tau)|d\tau\nonumber\\
&&+\int_0^{+\infty} |F^l(z_-+y_-)(\tau)-F^l(z_+(\tilde a,x_-,.))(\tau)|d\tau\Big).\label{5.26}
\end{eqnarray}
Hence from \eqref{03c} and \eqref{03d} it follows that
\begin{eqnarray}
&&|\Delta_{fc,1}|
\le{2n^{3\over 2}\sqrt{1-{|v_-|^2\over c^2}}}\Big({({\beta_1^l\over c}+2n^{1\over 2}\beta_2)(\|y_-\|_*+\delta)
\over (\alpha+1)({|v_-|\over 2^{3\over 2}}-r)(1+{|x_-|\over \sqrt{2}}-r)^{\alpha+1}}\label{5.27}\\
&&+{({2\beta_2\over c}+4n^{1\over 2}\beta_3^s)\|y_-\|_*
\over (\alpha+2)({|v_-|\over 2^{3\over 2}}-r)(1+{|x_-|\over \sqrt{2}}-r)^{\alpha+2}}
+{2n^{1\over 2}\beta_3^s\|y_-\|_*\over(\alpha+1)({|v_-|\over 2^{3\over 2}}-r)^2(1+{|x_-|\over \sqrt{2}}-r)^{\alpha+1}}\Big),\nonumber
\end{eqnarray}
where $\delta$ is defined in \eqref{5.18}.
We use the bound \eqref{5.25e} on $\delta$ and the bound \eqref{lc3} on $y_-=\A(y_-)$, and we use the estimate $r\le {1\over 2}$, and we obtain
\begin{equation}
|\Delta_{fc,1}|
\le{768n^4\beta^2\big(1-{|v_-|^2\over c^2}\big)({r\over c}+1)(1+{1\over c})(1+{1\over ({|v_-|\over 2^{3\over 2}}-r)})^2\over \alpha^2({|v_-|\over 2^{3\over 2}}-r)^2
({1\over 2}+{|x_-|\over \sqrt{2}})^{2\alpha+1}}.\label{D6}
\end{equation}
Then we use \eqref{D5} and \eqref{02c}, and we obtain
\begin{eqnarray}
|\Delta_{fc,2}|
&\le&{6n\big(1-{|v_-|^2\over c^2}\big)\over c}
\Big(\int_{-\infty}^0 |F^l(z_-)|(\tau)d\tau+\int_0^{+\infty}|F^l(z_+(\tilde a,x_-,.))|(\tau)d\tau\nonumber\\
&+&\int_{-\infty}^{+\infty}|F^s(z_-)(\tau)|d\tau\Big)\int_{-\infty}^{+\infty}|F^s(z_-) |(\tau)d\tau.\label{D9}
\end{eqnarray}
We use \eqref{03a}, \eqref{03b} and \eqref{5.3d}, and we have 
\begin{equation}
|\Delta_{fc,2}|\le{144n^2\beta^2\big(1-{|v_-|^2\over c^2}\big)\over c\alpha(\alpha+1)({|v_-|\over 2^{3\over 2}}-r)^2(1+{|x_-|\over \sqrt{2}})^{2\alpha+1}}.\label{D7}
\end{equation}
We use \eqref{02a} and an estimate similar to \eqref{3.9b}, and we have 
\begin{equation}
|\Delta_{fc,3}|\le{2^7n^{7\over 2}\beta_1^l(1-{|v_-|^2\over c^2})\over \alpha(\alpha+1)|v_-|^2(1+{|x_-|\over \sqrt{2}})^{2\alpha+1}}\big({\beta_2\over c}
+{2^{5\over 2}n^{1\over 2}\beta_3^s\over |v_-|}\big)\label{D8}
\end{equation}
Then we add the bounds of \eqref{D6}, \eqref{D7} and \eqref{D8}, and we obtain \eqref{t2f}.

We prove \eqref{t2g}. From  \eqref{6.1b} and $\tilde b_{sc}=\G_{v_-,x_-}(\tilde b_{sc})$ it follows that
\begin{eqnarray}
&&\tilde b_{sc}(v_-,x_-)-\Big(<\nabla g_j(g^{-1}(v_-)),\int_{-\infty}^0\int_{-\infty}^\sigma F^s(. v_-+x_-)(\tau)d\tau d\sigma\nonumber\\
&&-\int_0^{+\infty}\int_\sigma^{+\infty} F^s(. v_-+x_-)(\tau)d\tau d\sigma>\Big)_{j=1\ldots n}=\sum_{j=1}^8\Delta_{sc,j},\label{6.3t}
\end{eqnarray}
where
\begin{equation}
\Delta_{sc,1}:=\G_{v_-,x_-}(\tilde b_{sc})-\G_{v_-,x_-}(0),\ 
\Delta_{sc,2}:=\A(y_-)(0)-\A(0)(0),\ \label{6.3b}
\end{equation}
\begin{eqnarray}
\Delta_{sc,3}&:=&\int_{-\infty}^0\Big(<\int_0^1\nabla g_j\big(g^{-1}(v_-)+\int_{-\infty}^\sigma (F^l+\ep F^s)(z_-)(\tau)d\tau\big)d\ep\nonumber\\
&&-\nabla g_j(g^{-1}(v_-)),\int_{-\infty}^\sigma F^s(z_-)(\tau)d\tau>\Big)_{j=1\ldots n}d\sigma\label{6.4}
\end{eqnarray}
\begin{equation}
\Delta_{sc,4}:=\int_{-\infty}^0\big(<\nabla g_j(g^{-1}(v_-)),\int^\sigma_{-\infty}\big(F^s(z_-)(\tau)-F^s(. v_-+x_-)(\tau)\big)d\tau>\big)_{j=1\ldots n}d\sigma,\label{6.4c}
\end{equation}
\begin{eqnarray}
\Delta_{sc,5}&:=&-\int_0^{+\infty}\Big(<\int_0^1\Big(\nabla g_j\big(g^{-1}(\tilde a)-\int_\sigma^{+\infty}(\ep F(z_-+y_-)+(1-\ep)F^l(z_+(\tilde a,x_-,.)))(\tau))d\tau\big)\nonumber\\
&&-\nabla g_j(g^{-1}(v_-))\Big)d\ep,
\int_\sigma^{+\infty}F^s(z_-+y_-)(\tau)d\tau>\Big)_{j=1\ldots n}d\sigma,\label{6.3c}
\end{eqnarray}
\begin{equation}
\Delta_{sc,6}:=-\int_0^{+\infty}\Big(<\nabla g_j(g^{-1}(v_-)),\int_\sigma^{+\infty}\big(F^s(z_-+y_-)(\tau)
-F^s(z_-)(\tau)\big)d\tau>\Big)_{j=1\ldots n}d\sigma.\label{6.3d}
\end{equation}
\begin{equation}
\Delta_{sc,7}:=-\int_0^{+\infty}\Big(<\nabla g_j(g^{-1}(v_-)),\int_\sigma^{+\infty}\big(F^s(z_-)(\tau)
-F^s(.v_-+x_-)(\tau)\big)d\tau>\Big)_{j=1\ldots n}d\sigma,\label{6.3e}
\end{equation}
\begin{eqnarray}
\Delta_{sc,8}
&=&-\int_0^{+\infty}\Big(<\int_0^1\nabla g_j\big(g^{-1}(\tilde a)-\int_\sigma^{+\infty}(\ep F(z_-+y_-)
+(1-\ep)F^l(z_+(\tilde a,x_-,.)))(\tau)d\tau\big)\nonumber\\
&&\int_\sigma^{+\infty}\big(F^l(z_-+y_-)(\tau)-F^l(z_+(\tilde a,x_-,.))(\tau)\big)d\tau>\Big)_{j=1\ldots n}d\sigma.\label{6.2b}
\end{eqnarray}
From \eqref{lc3b} and \eqref{t2d} it follows that
\begin{equation}
|\Delta_{sc,1}|\le{8n^2\sqrt{1-{|v_-|^2\over c^2}}\max(\beta_1^l,\beta_2)({1\over c}+1)|\tilde b_{sc}|\over \alpha(\alpha+1)({|v_-|\over 2^{3\over 2}})^2({1\over 2}+{|x_-|\over \sqrt{2}})^\alpha}
\le{192n^4\left({1-{|v_-|^2\over c^2}}\right)\max(\beta_1^l,\beta_2)^2({1\over c}+1)^2
\over \alpha^2(\alpha+1)^2\big({|v_-|\over 2\sqrt{2}}-r\big)^4\big({1\over 2}+{|x_-|\over\sqrt{2}}\big)^{2\alpha}}.\label{E2}
\end{equation}
We use \eqref{4.27} ("$(f_1,f_2)=(y_-,0)$"), and we use the bound of $\|y_-\|_*=\|\A(y_-)\|_*$ given in \eqref{lc3} and $r\le{1\over 2}$, and we obtain 
\begin{equation}
|\Delta_{sc,2}|\le{192n^4(1-{|v_-|^2\over c^2})({1\over c}+1)({r\over c}+1)\max(\beta_1^l,\beta_2)\big(1+{1\over{|v_-|\over 2\sqrt{2}}-r})\over \alpha^2(\alpha+1)({|v_-|\over 2\sqrt{2}}-r)^3\big(1-r+{|x_-|\over \sqrt{2}}\big)^{2\alpha}}.\label{6.3f}
\end{equation}
The proof of the following estimates \eqref{E3}, \eqref{3.96}, \eqref{E5} and \eqref{6.3h} given below is similar to the proof of the estimates \eqref{3.23}, \eqref{3.18}, \eqref{3.16} and \eqref{3.17}
\begin{eqnarray}
|\Delta_{sc,3}|&\le&
{48n^3\max(\beta_1^l,\beta_2)(1-{|v_-|^2\over c^2})\over c\alpha^2(\alpha+1)({|v_-|\over 2^{3\over 2}})^3(1+{|x_-|\over \sqrt{2}})^{2\alpha}},\label{E3}\\
\max(|\Delta_{sc,4}|,|\Delta_{sc,7}|)&\le&{8n^3\beta^2(1-{|v_-|^2\over c^2}) \Big({1\over c}+{2n^{1\over 2}\over{|v_-|\over 2^{3\over 2}}}\Big)\over \alpha^2(\alpha+1)({|v_-|\over 2^{3\over 2}})^3(1+{|x_-|\over \sqrt{2}})^{2\alpha}},\label{3.96}\\
|\Delta_{sc,5}|&\le&{168n^3(1-{|v_-|^2\over c^2})\max(\beta_1^l,\beta_2)^2\over c\alpha^2(\alpha+1)({|v_-|\over 2\sqrt{2}}-r)^3(1-r+{|x_-|\over \sqrt{2}})^{2\alpha}},\label{E5}\\
|\Delta_{sc,6}|&\le&{96n^4\beta^2(1-{|v_-|^2\over c^2})(1+{1\over c})^2(1+{1\over {|v_-|\over 2^{3\over 2}}-r})^2\over \alpha^2(\alpha+1)({|v_-|\over 2^{3\over 2}}-r)^3(1+{|x_-|\over 2^{3\over 2}}-r)^{2\alpha}}.\label{6.3h}
\end{eqnarray}
From \eqref{D5}, \eqref{02a}, \eqref{03c} and \eqref{5.25e} it follows that
\begin{eqnarray}
|\Delta_{sc,8}|&\le&2n^{1\over 2}(1-{|v_-|^2\over c^2})^{1\over 2}\int_0^{+\infty}\int_\sigma^{+\infty}{({\beta_1^l\over c}+2n^{1\over 2}\beta_2)\delta d\tau d\sigma\over(1-r+{|x_-|\over \sqrt{2}}+({|v_-|\over 2^{3\over 2}}-r)\tau)^{\alpha+2}}\nonumber\\
&\le&{96n^4\beta^2(1-{|v_-|^2\over c^2})(1+{1\over c})^2(1+{1\over {|v_-|\over 2^{3\over 2}}-r})\over\alpha^2(\alpha+1)
({|v_-|\over 2^{3\over 2}}-r)^3({1\over 2}+{|x_-|\over 2^{3\over 2}})^{2\alpha}}.\label{E7}
\end{eqnarray}
Then we add the bounds on the right-hand sides of \eqref{E2}--\eqref{E7}, and we use \eqref{6.3t} and we obtain \eqref{t2g}.
\hfill$\Box$

\end{document}